\documentclass[12pt, draftclsnofoot, onecolumn]{IEEEtran}
\usepackage{graphicx}  % Written by David Carlisle and Sebastian Rahtz
\usepackage{enumitem}

\usepackage{subfigure} % Written by Steven Douglas Cochran
\usepackage{multirow}

\usepackage{amsmath}
\usepackage{amssymb}
\usepackage[justification=centering]{caption}
\usepackage{lipsum}
\usepackage{amsthm}
\usepackage{amsmath,bm}

\usepackage{algorithm}
\usepackage{algcompatible}
\algnewcommand\INPUT{\item[\textbf{Input:}]}%
\algnewcommand\OUTPUT{\item[\textbf{Output:}]}%

\usepackage{xcolor}
\usepackage{color}
\newtheoremstyle{mystyle}%                % Name
  {5mm}%                                     % Space above
  {5mm}%                                     % Space below
  {}%                                     % Body font
  {}%                                     % Indent amount
  {\bfseries}%                            % Theorem head font
  {:}%                                    % Punctuation after theorem head
  {3mm}%                                    % Space after theorem head, ' ', or \newline
  {}%                               % Theorem head spec (can be left empty, meaning `normal')
\theoremstyle{mystyle}

\newtheorem{theorem}{Theorem}
\newtheorem{lemma}{Lemma}
\newtheorem{corollary}{Corollary}
\newtheorem{definition}{Definition}

\newtheorem{proposition}{Proposition}

\newtheorem{example}{Example}

\definecolor{darkgreen}{rgb}{0.1,0.5,.5}
\definecolor{darkred}{rgb}{0.8,0,0}
\definecolor{teal}{rgb}{0.05,0.32,0.41}
\definecolor{darkblue}{rgb}{0,0.0,0.5}
\definecolor{blackgreen}{rgb}{0,0.4,0}
\definecolor{purple}{rgb}{0.5,0,0.3}
\definecolor{grey}{rgb}{0.7,0.5,0.5}
\definecolor{orange}{rgb}{0.6,0.4,0.1}

\newcommand{\bfit}[1]{ {\textit{\textbf{#1}}}}
\newcommand{\rom}[1]{\uppercase\expandafter{\romannumeral #1\relax}}

\begin{document}
\title{Analysis of Maximal Topologies Achieving Optimal DoF and DoF $\frac{1}{n}$ in Topological Interference Management}
\author{Jong-Yoon Yoon and Jong-Seon No
\thanks{J.-Y Yoon and J.-S. No are with the Department of Electrical and Computer
Engineering, INMC, Seoul National University, Seoul 08826, Korea. e-mail:
yjy998@ccl.snu.ac.kr, jsno@snu.ac.kr.}
}

% note the % following the last \IEEEmembership and also the first \thanks -
% instance, "A\textbf{} \textbf{}B" will typeset as "A B" not "AB". If you want
% "AB" then you have to do: "A\textbf{}\textbf{}B"
\markboth{Submission for...}{Regular Paper}

% If you want to put a publisher's ID mark on the page
% (can leave text blank if you just want to see how the
% text height on the first page will be reduced by IEEE)
%\pubid{0000--0000/00\$00.00~\copyright~2002 IEEE}

%\specialpapernotice{(Invited Paper)}

\maketitle

\begin{abstract}

Topological interference management (TIM) can obtain degrees of freedom (DoF) gains with no channel state information at the transmitters (CSIT) except topological information of network in the interference channel. It was shown that TIM achieves the optimal symmetric DoF when internal conflict does not exist among messages [6]. However, it is difficult to assure whether a specific topology can achieve the optimal DoF without scrutinizing internal conflict, which requires lots of works. Also, it is hard to design a specific optimal topology directly from the conventional condition for the optimal DoF.
With these problems in mind, we propose a method to derive maximal topology directly in TIM, named as alliance construction in $K$-user interference channel. That is, it is proved that a topology is maximal if and only if it is derived from alliance construction. We translate a topology design by alliance construction in message graph into topology matrix and propose conditions for maximal topology matrix (MTM). 
Moreover, we propose a generalized alliance construction that derives a topology achieving DoF $1/n$ for $n\geq3$ by generalizing sub-alliances. A topology matrix can also be used to analyze maximality of topology with DoF $1/n$.
\end{abstract}
\begin{IEEEkeywords}
Alliance, alliance construction, degrees-of-freedom (DoF),  internal conflict, maximal topology matrix (MTM), topological interference management (TIM).
\end{IEEEkeywords}

\vspace{10pt}
\section{Introduction}
Recently, there have been many remarkable advances in the wireless networks with interference and the most astonishing breakthrough is the idea of interference alignment (IA) [1]. IA is a scheme to design signals in such a way that interference signals can be overlapped and separated from desired signal at each receiver so that each receiver can recover its desired message with gains of degrees of freedom (DoF). IA made a boom to research interference channel and many related studies based on IA have been done. The initial researches on the IA mainly depend on the perfect and instantaneous channel state information at the transmitters (CSIT) [2], [3].

However, perfect CSIT assumption is not practical and challenging, because perfect CSIT is rarely available to transmitter. Also, when the number of users is large or the channel changes rapidly, the burden of CSIT becomes large. Considering the difficulty of perfect CSIT, researchers have embarked on exploring settings with relaxed CSIT assumptions. It is shown that the setting with delayed CSIT and reconfigurable antenna can achieve the optimal symmetric DoF, which is the same as perfect CSIT assumption model [4]. If the fading channels of different users follow some structured patterns, then blind IA could improve DoF beyond the absolutely no CSIT case [5]. 

Nevertheless, most of the studies are based on the theoretical insights, which remain fragile so far to be applied to practice directly. Also, these traditional interference management schemes based on IA always consider all interference links regardless of their strength, which results in unnecessary waste on resources such as time and antennas. As the strength of interference rapidly decays with distance due to shadowing, blocking, an path loss, interference from some sources is necessarily weaker than others, which is enough to be ignored. There are more opportunities in terms of DoF and resources by utilizing the characteristic of partial connectivity in actual channels.

With the more practical assumptions of interference channel escaped from the pessimistic view and relaxation for heavy CSIT assumptions, interference management with no channel state information except the knowledge of the connectivity at the transmitters has been suggested under the name of the "topological interference management (TIM)"[6]. Jafar suggested that index coding problem could be applied to TIM problem only with linear solutions and translated the index coding problem into TIM problem in a way of analyzing DoF gains [6]. It has been shown that under the topology satisfying certain conditions, TIM can obtain gains in terms of DoF and further achieve one half DoF per user, which is optimal for an interference channel with perfect CSIT. And it can be achieved with only topological information. 

Inspired by the new framework of topological interference management that has a merit of tremendous reduction of CSIT, there have been a lot of follow-up researches in line with various assumptions such as channel, antenna, cellular network, transmit cooperation, and message passing. Fast fading channel [7] and alternating connectivity [8] are also considered and fundamental limits on multiple antennas in the TIM setting is derived [9]. Furthermore, TIM is studied in the downlink cellular network with hexagonal structure [10] and more gains of DoF is achievable with the help of message passing in uplink cellular network [11].

Unlike above follow-up studies of TIM, we further develop the research in \cite{TIM} in more practical sense rather than changing assumptions or putting some schemes which help to enhance DoF. Since the study on TIM in \cite{TIM} is based on index coding problems which mainly focus on relationship among messages,
it is hard to design a specific optimal topology directly from the optimal DoF conditions in TIM,  which are not suitable for dealing with actual network topologies. For this problem, we raise a question, "Is it possible to find and derive all topologies achieving the optimal symmetric DoF in TIM?" This is the motivation of our research. In order to avoid finding unnecessary topologies, which are sub-topologies of other topologies, we focus on finding only maximal topology with the optimal DoF $1/2$, where any interference link cannot be added without breaking the optimality. 

In this paper, we reinterpret the previous conditions of messages for the optimal DoF into more understandable conditions of messages by introducing the subset of messages with constraints, called \textit{alliance} and propose how to construct a maximal topology. With aid of the alliance construction, all maximal topologies can be derived. Meanwhile, there is still unsolved question, "How can we determine whether certain topology is maximal or not?" In this paper, we answer these questions with the alliance construction and analysis of topology in matrix perspective.

More specifically, our contributions in this paper are summarized as follows:

\begin{itemize}
    \item We propose the  alliance construction which derives maximal topology by stipulating several conditions for constructing alliances. It is proved that a topology is maximal if and only if it is derived from the alliance construction. Properties of alliance construction are introduced such as the maximum number of alliances to be constructed for the given number of messages $K$ and the partitioning method of messages into sub-alliances. 
    
    \item Message relationship based on alliance construction is translated into topology matrix in TIM. Permutation of the topology matrix is used to demonstrate the characteristics of the alliances easily in the topology matrix. The conditions for maximal topology matrix (MTM) are characterized and the discriminant of topology matrix for maximality and transformation of non-MTM into MTM are proposed.
    
    \item Alliance construction is generalized by introducing generalized sub-alliances, which extends the range of DoF values of maximal topology from alliance construction. The conditions for MTM with DoF $1/n$, the discriminant and, transformation of non-MTM into MTM with DoF $1/n$ are also proposed
\end{itemize}

The rest of this paper is organized as follows: Section \rom{2} presents the system model and definitions for TIM. The main results of this paper, that is, alliance construction for maximal topology with optimal DoF $1/2$ is proposed in Section \rom{3}. The topology analysis in matrix perspective is discussed in Section \rom{4}. The generalized alliance construction and the  maximal topology with DoF $1/n$ are proposed in Section \rom{5} and Section \rom{6} concludes the paper.

%Section \rom{6} discusses that alliance construction can give a criterion for interference links in TIM-TIN decomposition, 

\vspace{10pt}
\section{System Model and Preliminaries} \label{sec_preliminaries}

Throughout the paper, some notations are defined as follows. $\mathrm{\mathit{A}}$, $\bfit{A}$, and $\mathcal{A}$ represent a variable, a matrix, and a set, respectively. $\left| \mathcal{A} \right|$ denotes the cardinality of the set $\mathcal{A}$ and $a_{ij}$ is the $(i,j)$-th entry of the matrix $\bfit{A}$. Let $\mathcal{K} =\{ 1,2,\cdots,K\}.$  $\mathbf{1}$ and $\mathbf{0}$ represent all-one and all-zero vectors, respectively.

\subsection{Channel Model}

We consider the TIM setting \cite{TIM} in a partially connected $K$-user interference channel, where $K$ transmitters want to send $K$ independent messages to $K$ receivers equipped with a single antenna. 
Then, the received signal at the receiver $i$ through partially connected channel at time instant $t$ is represented as
\begin{equation}\label{equ_rec_1st}
\textit{Y}_{i}(t)=\sum_{j\in \mathcal{S}_i}{\textit{h}_{ij}(t)\textit{X}_{j}(t)}+\textit{Z}_{i}(t), \text{ } i\in \mathcal{K},
\vspace{10pt}
\end{equation}

\noindent where $X_j(t)$ is the transmitted signal with the average power constraint $\mathbb{E}[\textit{X}^2_{j}(t)]\leq P$, ${Z_{i}(t)}$ is the Gaussian noise with zero-mean and unit variance, ${h_{ij}(t)}$ is the channel coefficient between transmitter ${j}$ and receiver ${i}$, and $\mathcal{S}_{i}$ represents a set of the indices of transmitters that are connected to receiver $i$. The network topology is denoted by $\mathcal{T}$, which is directed bipartite graph with transmitters and receivers at each side, and with edges from transmitters to receivers only when they are connected.

Similar to TIM researches in \cite{TIM}-\cite{TIM_MP}, the following channel state information (CSI) is assumed:
\begin{enumerate}[label=(\roman*)]
    \item The channel coefficients are assumed to be fixed throughout the duration of communication such that $h_{ij}(t)=h_{ij}$ and thus the network topology $\mathcal{T}$ is also assumed to be fixed.
    
    \item The channel coefficients $h_{ij}$ for all $i, j$ are unavailable at the transmitters, but the network topology $\mathcal{T}$ is known to all transmitters and receivers.
    
    \item The channel state information at the receiver (CSIR) includes only the information of the desired channel coefficient $h_{ii}$ at each receiver.
\end{enumerate}

\subsection{Problem Statement}

In this paper, we use some definitions for the TIM problem in \cite{TIM}. For the theorem for optimal DoF in TIM, the definitions of alignment graph, conflict graph, alignment set, and internal conflict are given as follows.

\begin{definition}[Alignment graph] It is called an alignment graph if messages $W_{i}$ and $W_{j}$ in the graph are connected with solid black edge whenever the source(s) of both these messages are heard by a receiver that desires message $W_{k}$, $k\neq i$ and $k\neq j$.
\end{definition}

\begin{definition}[Conflict graph] It is called a conflict graph if each message $W_{i}$ is connected by a dashed red edge to all other messages whose sources are heard by a destination that desires message $W_{i}$.
\end{definition}

\begin{definition}[Alignment set] Each connected component of an alignment graph is called an alignment set.
\end{definition}

\begin{definition}[Internal conflict] If two messages that belong to the same alignment set have a conflict edge between them, it is called an internal conflict.
\end{definition}
 
 \begin{definition}[Maximal topology]
A topology is maximal if there is no internal conflict in the alignment set(s) and any interference link cannot be added without occurring internal conflict. The maximal topology for $K$-user interference channel is denoted as $\mathcal{T}_{K}^{\mathrm{M}}$.
\end{definition}

\begin{definition}[Topology matrix]
A topology matrix for $K$-user interference channel, $\bfit{T}_{K}=[t_{ij}]_{K\times K}$ is defined as $t_{ij}=1,$ if there is a link between transmitter $j$ and receiver $i$, and $t_{ij}=0,$ otherwise.
\end{definition}

%(i.e,, the DoF which can be achieved by all users simultaneously)
In this paper, we use the strategy for DoF using linear beamforming scheme in \cite{TIM} and set the DoF as our main performance criterion.

\begin{definition} [Degrees of freedom]
The DoF  or multiplexing gain $d$ of the channel is defined as 
\begin{equation*}
d=\lim_{SNR \to \infty}{\frac{R(SNR)}{log (SNR)} },
\end{equation*}
where $R(SNR)$ is an achievable rate at the $SNR$.
\end{definition}

\section{Alliance Construction and Maximal Topology}

In this section, we propose our main results, alliance construction and maximal topology set derived from alliance construction in TIM. The optimal DoF conditions in TIM are already proposed in [6]. However, since they are focused on the message relationship rather than transceivers or topology of interference channel, it is not easy to check whether a given interference channel achieves the optimal symmetric DoF or not without drawing alignment and conflict graphs and investigating the existence of internal conflict, which requires lots of works. In other words, Theorem 1 in [6] does not directly produce a topology that achieves the optimal DoF. This is the beginning of our study and one of our main contributions is to derive all topologies achieving optimal symmetric DoF by combining and reinterpreting the alignment set and the internal conflict into a single concept, called \textit{alliance}. To this end, we define a maximal topology in the previous section. We only consider maximal topology because any non-maximal topology is sub-topology of maximal topology. Once alliances are constructed for all messages, the relationship among messages is determined naturally, which derives a maximal topology.

In Subsection \rom{2}.A, we propose alliance construction and in Subsection \rom{2}.B, a linear beamforming scheme for messages in each alliance is proposed, which shows how it achieves the optimal symmetric DoF in TIM. In Subsection \rom{2}.C, we derive the maximum number of alliances for the given number of messages $K$ and the partitioning method of messages into sub-alliances. In Subsection \rom{2}.D, the discriminant of maximal topology is proposed using alliance construction and the transformation of non-maximal topology into maximal one is also proposed.

\subsection{Alliance Construction}

In the alignment graph and conflict graph, each message implicitly represents a pair of transmitter and receiver which sends and wants it, where DoF is analyzed based on the relationship of messages rather than the network topology. In general, it is difficult to derive the optimal topology in terms of DoF from the message-oriented graph. Theorem 4 in \cite{TIM} tells that the topology can achieve the optimal symmetric DoF in TIM if and only if  its alignment sets do not have any internal conflict. Since we only focus on the maximal topology, it is better to consider topology in terms of the alignment set with no internal conflict rather than messages themselves.  However, it is not enough to derive maximal topology because maximal topology has not only the optimal DoF but also maximality of interference links with maintaining the optimality. Thus, the alignment set with constraints is needed, that is, deconflicting of messages within the alignment sets and satisfying maximality of topology. In order to propose how to design a maximal topology in TIM, we introduce alliance of messages as follows.

%Thus, we need more intuitive conditions to convert message relationship into network topology with the optimal DoF. To this end, the alignment set and internal conflict need to be interpreted into more topology-oriented concept. 
%the alignment set of connected  messages in alignment graph, where directly connected message commonly conflict with other messages(s). 

\begin{definition}[Alliance]
An alliance is a partition set of whole messages in the message graph of TIM satisfying the following conditions.
\begin{enumerate}[label=(\roman*)]
    \item (Deconflict) All messages in the set deconflict each other.
    \item (Cooperative conflict) All messages in the set conflict with each message in a subset of other messages.
\end{enumerate}
\end{definition}

According to the above definition, there are two conditions for messages in alliance. The first condition is the deconflict among messages in alliance, which prevents internal conflict in alignment sets. The second condition means that if a single message $W_i$ in alliance conflicts with message $W_k$ in other alliances, then the other messages in the same alliance must conflict with $W_k$. Thus, all messages in alliance are fully connected by alignment edges. The cooperative conflict makes the topology include as many interference links as possible without changing message relationship, that is, it is necessary condition for a maximal topology.
%which adds the interference links between transmitters of the other messages and receiver who wants $W_k$ to the topology of network. 

\begin{lemma}[Cooperative conflict]
	If a topology is maximal, all messages in alignment set conflict with each message in a subset of other messages.
\end{lemma}

\begin{proof}
Suppose that the messages in an alignment set do not satisfy cooperative conflict, that is, messages in a subset of the alignment set conflicts with $W_k$. Consider the possible  relationship of the remaining messages in the alignment set and $W_k$ without occurring internal conflict. Due to internal conflict, the remaining messages and $W_k$ can be connected not with  alignment edges but with conflict edges. Since the messages in the subset and the remaining messages are already in the same alignment set, connecting the remaining messages and $W_k$ with conflict edges does not change the message relationship. Thus, the topology is not maximal and we prove it.
\end{proof}
The cooperative conflict of alliance is summarized as in the following sentence. "The enemy of my friends (messages in the alliance) is also my enemy."

\begin{figure*}[t]
\centering
\subfigure[Alliance]{\includegraphics[width=0.43\linewidth]{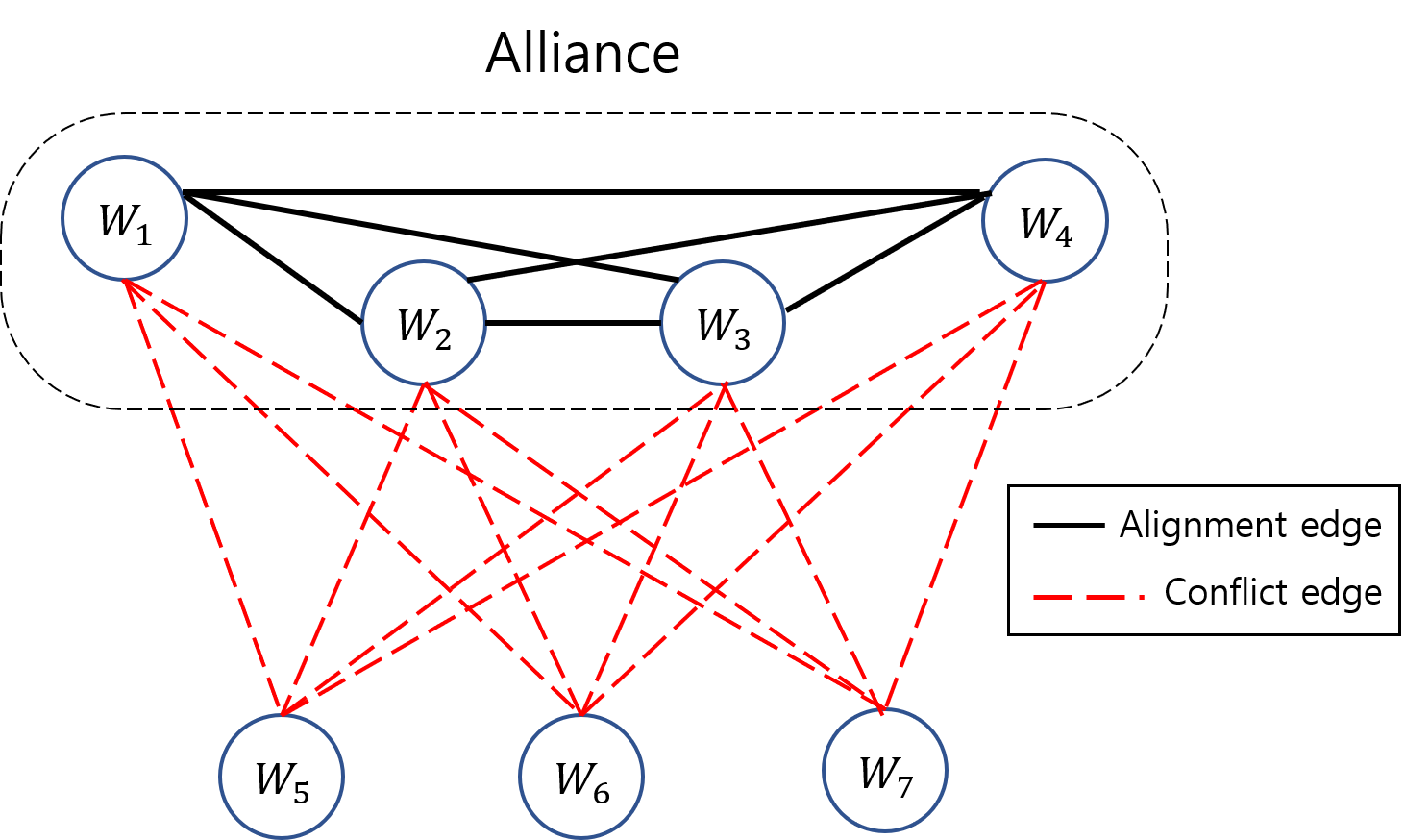}}
\hspace{25pt}
\subfigure[Alignment set]{\includegraphics[width=0.35\linewidth]{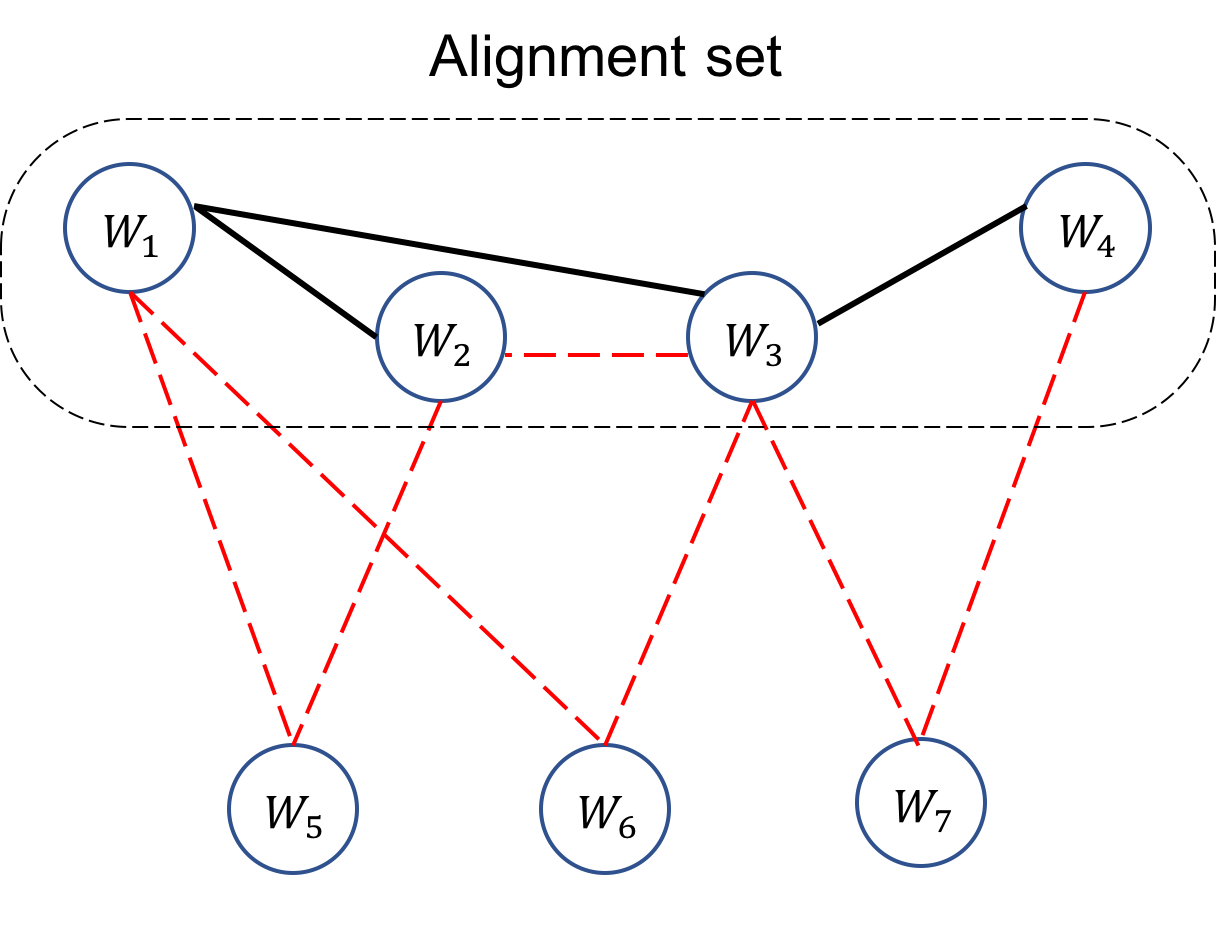}}
\caption{The differences between alliance and alignment set.}
\label{MTMF}
\end{figure*}

We describe difference between alliance and alignment set in Fig. 1. The alignment set can have internal conflict in the set and all messages in the set do not have to conflict with each of other messages. On the contrary, messages in alliance follow the deconflict of messages in alliance, which prevents internal conflict, and the cooperative conflict, which follows conditions for maximality of topology. That is, alliance is an alignment set that satisfies the cooperative conflict with no internal conflict. It is more suitable to analyze a maximal topology by the alliance than the alignment set. However, the alliance itself is not enough for deriving network topology. The alliance is just a subset of messages, which determines relationship of messages in the alliance and its conflicting messages. Now, it is needed to establish inter-alliance relationship in order to complete relationship of whole messages in the message graph, which is called \textit{construction of alliances}. Here, we need some definitions.

\begin{definition}[Hostility and mutual hostility]
Let $\mathcal{A}_i$ represent the $i$th alliance in the message graph. $\mathcal{A}_i$ is said to be hostile to $\mathcal{A}_j$ if and only if all messages in $\mathcal{A}_i$ conflict with all messages in $\mathcal{A}_j$, denoted by
   \begin{equation}
   		\mathcal{A}_i \rightarrow \mathcal{A}_j.
   \end{equation}

Also, $\mathcal{A}_i$ and  $\mathcal{A}_j$ are said to be mutually hostile if and only if all messages in $\mathcal{A}_i$ conflict with all messages in $\mathcal{A}_j$ and vice versa, denoted by
   \begin{equation}
   		\mathcal{A}_i  \Longleftrightarrow \mathcal{A}_j.
   \end{equation}
\end{definition}

The possible number of alliances is limited to two if we assume mutual hostility of all alliances as in the following lemma.

\begin{lemma}[Mutual hostility]
	If all alliances are mutually hostile, there exist only two alliances.
\end{lemma}

\begin{proof}
Suppose that there are three alliances $\mathcal{A}_1$, $\mathcal{A}_2$, and $\mathcal{A}_3$ with mutual hostility. Due to the mutual hostility of all alliances, $\mathcal{A}_1$ and $\mathcal{A}_2$ are hostile to $\mathcal{A}_3$. This is contradiction that $\mathcal{A}_1$ and $\mathcal{A}_2$  should be combined into a single alliance by cooperative conflict, but cannot be combined due to hostility between them. Similarly, more than three alliances cannot exist with the mutual hostility of all alliances. Therefore, there exist only two alliances if all alliances are mutually hostile.
\end{proof}

The mutual hostility can be related to maximality of topology.

\begin{theorem}[2-alliance construction]
Suppose that there exist two alliances $\mathcal{A}_1$ and $\mathcal{A}_2$ in the message graph of $K$-user interference channel.  A topology derived from 2-alliance construction is maximal if and only if $\mathcal{A}_1$ and $\mathcal{A}_2$ are mutually hostile
\begin{equation}
\mathcal{A}_1 \Longleftrightarrow \mathcal{A}_2.
\end{equation}
 The achievable symmetric DoF of the topology from 2-alliance is optimal as
 \begin{equation}
d_{sym} = \frac{1}{2}.
\end{equation}

\end{theorem}

\begin{proof}
(Necessary) Assume that $\mathcal{A}_1$ and $\mathcal{A}_2$ are not mutually hostile, that is, all messages in $\mathcal{A}_1$ conflict with some messages in $\mathcal{A}_2$ or all messages in $\mathcal{A}_2$ conflict with some messages in $\mathcal{A}_1$. Then, we can add conflict edges between messages in $\mathcal{A}_1$ and $\mathcal{A}_2$ without occurring the internal conflict and thus it is not maximal. (Sufficient) From Lemma 2, there are only two alliances if all alliances are mutually hostile. Then all messages in $\mathcal{A}_1$ are fully conflict with all messages in $\mathcal{A}_2$ and vice versa. Also, it is not possible to add any conflict edges among messages in the same alliance due to the deconflict of messages in alliance. Thus the topology is maximal.
\end{proof}

We omit the proof of DoF achievability because it has already shown in [6] for the alignment graph and conflict graph and the achievable scheme is proposed in the following subsection.

\begin{figure*}[t]
\centering
\subfigure[A maximal topology]{\includegraphics[width=0.21\linewidth]{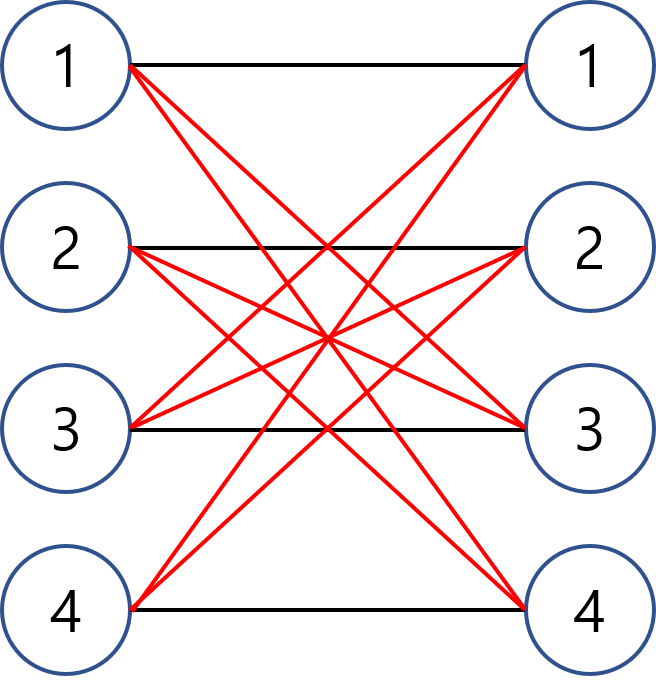}}\label{MT41}
\hspace{10pt}
\subfigure[Alignment-conflict graph of (a)]{\includegraphics[width=0.23\linewidth]{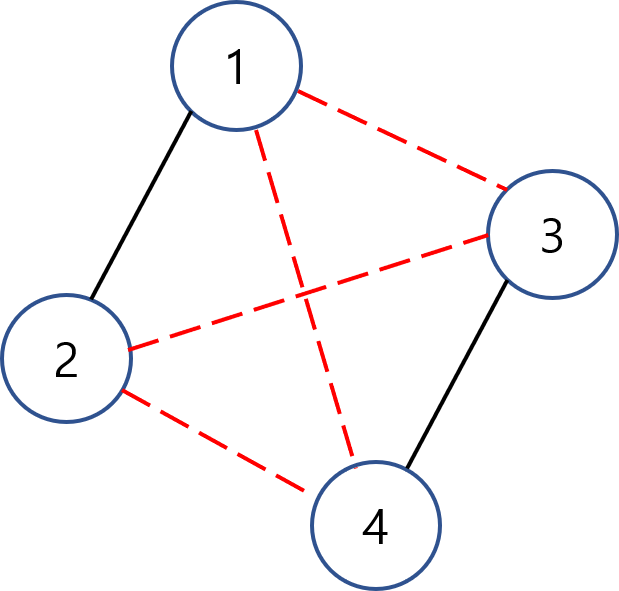}}\label{MT4_2}
\hspace{10pt}
\subfigure[Another maximal topology]{\includegraphics[width=0.21\linewidth]{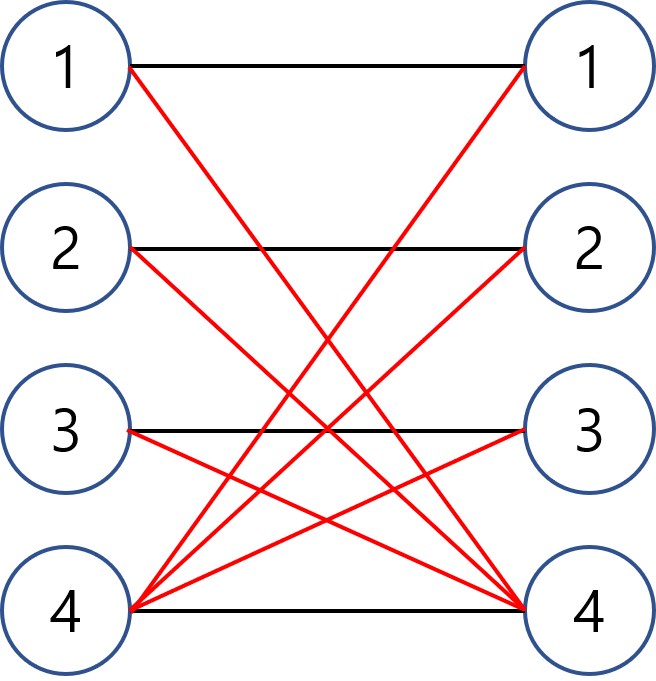}}\label{MT4_3}
\hspace{10pt}
\subfigure[Alignment-conflict graph of (c)]{\includegraphics[width=0.23\linewidth]{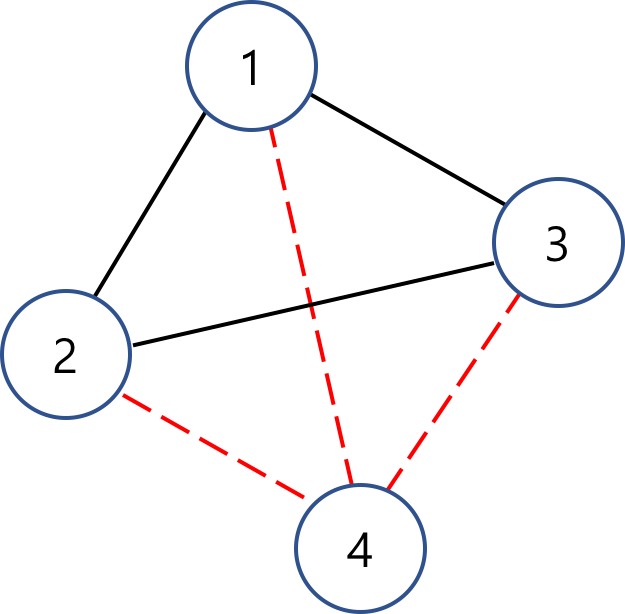}}\label{MT4_4}
\caption{Two maximal topologies for 4-user interference channel and related graphs.}
\label{MT4}
\end{figure*}

\begin{example}
In Fig. \ref{MT4}, there are two maximal topologies for 4-user interference channel. They are maximal topologies derived from 2-alliance construction. The solid black edges indicate that two messages are connected as alignment edge and belong to the same alliance (alignment set). The dashed red edges show conflict between them. In Fig. 2 (a), there are two alliances, $\mathcal{A}_1=\{W_1, W_2\}$ and $\mathcal{A}_2=\{W_3, W_4\}$. The interference channel in Fig. 2 (c) is another maximal topology, where $\mathcal{A}_1=\{W_1, W_2,W_3\}$ and $\mathcal{A}_2=\{W_4\}$.
\end{example}

The two topologies in Example 1 are all possible maximal topologies derived from 2-alliance construction for 4-user interference channel if we do not take into account the indices of messages. However, there are other maximal topologies which are generated from other than 2-alliance construction. The natural question is, "Is there other way to satisfy maximality with giving up mutual hostility of all alliances?" The alliance construction can be generalized by setting hostility of alliances in a more general way. The key idea is to construct alliances, where each subset of messages in the alliance are interfered separately from all messages of each alliance. To this end, some definitions are needed as follows.

\begin{definition}[Sub-alliance] Let $\mathcal{A}_i$ and $\mathcal{A}_j$ be alliances. $\mathcal{A}_{i,j}$ and $\mathcal{A}_{i,k}$ are partition sets of $\mathcal{A}_i$, $\mathcal{A}_{i,j}\cap \mathcal{A}_{i,k}=\emptyset$ for distinct $i$, $j$, and $k$ and $\bigcup_{j\neq i}{\mathcal{A}_{i,j}}=\mathcal{A}_i$. Then $\mathcal{A}_{i,j}$ represents sub-alliance of $\mathcal{A}_i$, where all messages in $\mathcal{A}_j$ conflict with each message in $\mathcal{A}_{i,j}$.
\end{definition}

\begin{definition}[Partial hostility] $\mathcal{A}_j$ is said to be partially hostile to $\mathcal{A}_i$ if all messages in $\mathcal{A}_j$ only conflict with each message in $\mathcal{A}_{i,j}$ of $\mathcal{A}_{i}$, denoted as
   \begin{equation}
   		\mathcal{A}_j \rightharpoonup \mathcal{A}_i
   \end{equation}
\begin{equation}
   		\text{identically}\,\mathcal{A}_j \rightarrow \mathcal{A}_{i,j}.
   \end{equation}
\end{definition}

\begin{definition}[Mutually partial hostility] $\mathcal{A}_i$ and $\mathcal{A}_j$ are said to be mutually partial hostile if all messages in $\mathcal{A}_i$ only conflict with each message in $\mathcal{A}_{j,i}$ of $\mathcal{A}_{j}$ and all messages in $\mathcal{A}_j$ only conflict with each message in $\mathcal{A}_{i,j}$ of $\mathcal{A}_{i}$, where at least one of  $\mathcal{A}_{i,j}$ and $\mathcal{A}_{j,i}$ are non empty sets, denoted by
   \begin{equation}
   		\mathcal{A}_i \rightleftharpoons \mathcal{A}_j
   \end{equation}	
\begin{equation}
   		\text{identically}\,\mathcal{A}_i \rightarrow\mathcal{A}_{j,i}\,\text{ and/or }\,\mathcal{A}_{j} \rightarrow \mathcal{A}_{i,j}.
   \end{equation}
\end{definition}

Even though one of sub-alliances $\mathcal{A}_{i,j}$ and $\mathcal{A}_{j,i}$ is an empty set, $\mathcal{A}_i$ and $\mathcal{A}_j$ are also said to be mutually partial hostile in this paper. That is, $\mathcal{A}_i \rightleftharpoons \mathcal{A}_j$ includes three cases that only $\mathcal{A}_i$ is partially hostile to $\mathcal{A}_j$ or only $\mathcal{A}_j$ is partially hostile to $\mathcal{A}_j$ or both $\mathcal{A}_i$ and $\mathcal{A}_j$ are partially hostile to each other. If there is no hostility between any two alliances, two alliances are combined by adding conflict edges without internal conflict and $N$ alliances are reduced to $N-1$ alliances.

\begin{lemma}\label{aml} 
If sub-alliances $\mathcal{A}_{i,j}$ and $\mathcal{A}_{j,i}$ are all empty sets,  the topology is not maximal.
%% and by transformed into $N-1$. 
\end{lemma}

\begin{proof}Since there are no conflict edges between messages in $\mathcal{A}_i$ and $\mathcal{A}_j$, two alliances can be merged into an alliance. If we merge them, the merged alliance should follow the cooperative conflict of alliance. However, conflict edges between messages in $\mathcal{A}_{k,i}$ and $\mathcal{A}_{j}$ and conflict edges between messages in $\mathcal{A}_{k,j}$ and $\mathcal{A}_{i}$ can be added for all distinct $i$, $j$, and $k$, which means that the topology is not maximal.\end{proof}
%The merging and adding conflict properly leads to $N-1$ alliance construction satisfying maximality. 

%\begin{figure*}[t]
%\centering
%subfigure[Inappropriate 3-alliance construction]{\includegraphics[width=0.26\linewidth]{L1_1.png}}
%\hspace{25pt}
%\subfigure[2-alliance construction]{\includegraphics[width=0.31\linewidth]{L1_2.png}}
%\caption{Alliance merging}
%\label{AM}
%\end{figure*}

%\begin{example}
% Since there is no hostility between $\mathcal{A}_{1}$ and $\mathcal{A}_{3}$ in Fig. \ref{AM}, alliances are constructed improperly and the topology is not maximal. In order to correct the construction, $\mathcal{A}_{1}$ and $\mathcal{A}_{3}$ should be merged into an alliance and all possible interference edges are inserted to the topology in order to satisfy maximality. We discuss transformation of  non-maximal topology into maximal topology at Subsection E.
%\end{example}

Using  the sub-alliances and the mutually partial hostility, Theorem 1 can be modified into the following theorem.
%It implies that $\mathcal{A}_{ij}$ and $\mathcal{A}_{ji}$ are all empty sets means that  they are not hostile and no need to distinguish them.
%Assuming  mutually partial hostility of all alliances, following Theorem summarizes the alliance construction. 

\begin{theorem}[$N$-alliance construction]
Suppose that there exist $N$ alliances for $K$-user interference channel.  A topology derived from $N$-alliance construction is maximal if and only if any $\mathcal{A}_i$ and $\mathcal{A}_j$ are mutually partial hostile, that is, 
\begin{equation}
\mathcal{A}_i \rightleftharpoons \mathcal{A}_j, \,\text{for all}\, i,j.
\end{equation}
 Further, the achievable symmetric DoF is optimal as
 \begin{equation}
d_{sym} = \frac{1}{2}.
\end{equation}
\end{theorem}

\begin{proof}(Necessary) Assume that for some $i$ and $j$, $\mathcal{A}_i$ and $\mathcal{A}_j$ are not mutually partial hostile, where there are three cases: i) $\mathcal{A}_{i,j}=\mathcal{A}_{j,i}=\emptyset$, ii) $\mathcal{A}_{i,j}\cap\mathcal{A}_{i,k}\neq\emptyset$, iii) $\cup_{j}\mathcal{A}_{i,j}\neq\mathcal{A}_i$. From Lemma 3, we have already proved the case i). For the second case, since messages in $\mathcal{A}_j$ and $\mathcal{A}_k$ has common messages to conflict with messages in $\mathcal{A}_i$, messages in $\mathcal{A}_j$ and $\mathcal{A}_k$ should be combined into an alignment set. But there exist conflicts between messages in $\mathcal{A}_{j,k}$ and $\mathcal{A}_{k}$ or messages in $\mathcal{A}_{k,j}$ and $\mathcal{A}_{j}$, which means that internal conflicts exist and the topology is not maximal. For the third case, there are at least one messages in $\mathcal{A}_{i}$, which is not interfered with. Then, some interference links can be added and thus the topology is not maximal.

(Sufficient) Assume that $N$ alliances are mutually partial hostile. Due to the mutually partial hostility for all pairs of alliances, there exist conflicts between all messages in $\mathcal{A}_i$ and $\mathcal{A}_{j,i}$ and between all messages in $\mathcal{A}_j$ and $\mathcal{A}_{i,j}$ for any distinct $i$ and $j$.  Let us add a conflict edge from a message $W_e$ in $\mathcal{A}_{i}$ to $W_p$ in $\mathcal{A}_{j,k}$ for any distinct $i$, $j$, and $k$. Then all messages in $\mathcal{A}_{k,i}$ and $W_e$ are connected with alignment edges, which results that all messages in $\mathcal{A}_i$ and $\mathcal{A}_k$ are tied as an alignment set. Since $\mathcal{A}_i$ and $\mathcal{A}_k$ are hostile to each other, there exist internal conflict in the alignment set $\mathcal{A}_i \cup \mathcal{A}_k$. Thus, internal conflict always occurs by adding a conflict edge between any two messages and thus the topology is maximal. Similarly, we omit the proof of DoF achievability.\end{proof}

The following corollary can be stated without proof.

\begin{corollary}
A topology is maximal if and only if all distinct alignment sets are alliances with mutually partial hostility.
\end{corollary}

Let $\mathcal{A}$ be a set of alliances given as $\mathcal{A}=\{\mathcal{A}_{1},\mathcal{A}_{2},\cdots,\mathcal{A}_{N}\}$. Here, alliance $\mathcal{A}_{i}$ is partitioned into sub-alliances $\mathcal{A}_{i,j} $, $1\leq j \leq N$, $j\neq i$. Let $\mathcal{A}_{sub}$ be a set of sub-alliances given as $\mathcal{A}_{sub}=\{\mathcal{A}_{i,j}\mid 1 \leq i,j \leq N, i\neq j \}$. Then, the sub-alliance graph is defined as:

\begin{definition}[Sub-alliance graph] Suppose that there are $N$ alliances. A directed graph $\mathcal{G}=(\mathcal{A}_{sub},\mathcal{I})$ is called a sub-alliance graph, where there exist directed edges from all sub-alliances in $\mathcal{A}_i$ to $\mathcal{A}_{j,i}$ for each pair of $i$ and $j$.
%if $\mathcal{A}_{ij}$ is not empty set.
\end{definition}

\begin{proposition}[Topology derivation] A maximal topology $\mathcal{T}_{K}^{\mathrm{M}}$ is derived from sub-alliance graph as follows:
\begin{enumerate}[label=(\roman*)]
\item There exists a direct link between transmitter $m$ and receiver $m$ for each message $W_m$.
\item There exists an interference link between transmitter $m$ whose message belongs to alliance $\mathcal{A}_{i}$ and receiver $n$ whose message belongs to $\mathcal{A}_{j,i}$ for all distinct $i$ and $j$.
\end{enumerate}
\end{proposition}

\begin{figure*}[t]
\centering
\subfigure[Sub-alliance graph]{\includegraphics[width=0.35\linewidth]{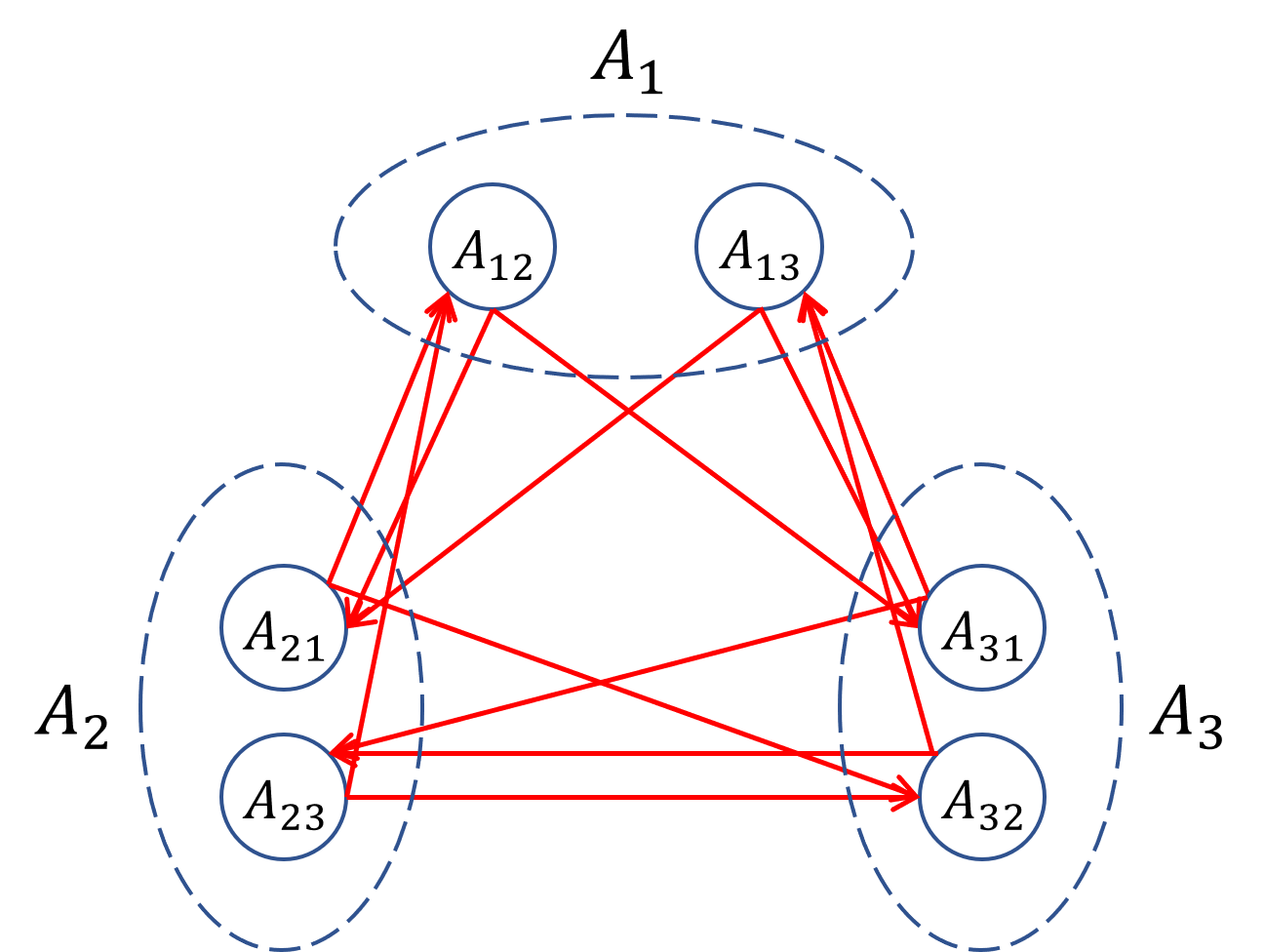}}
\hspace{13pt}
\subfigure[Maximal topology]{\includegraphics[width=0.2\linewidth]{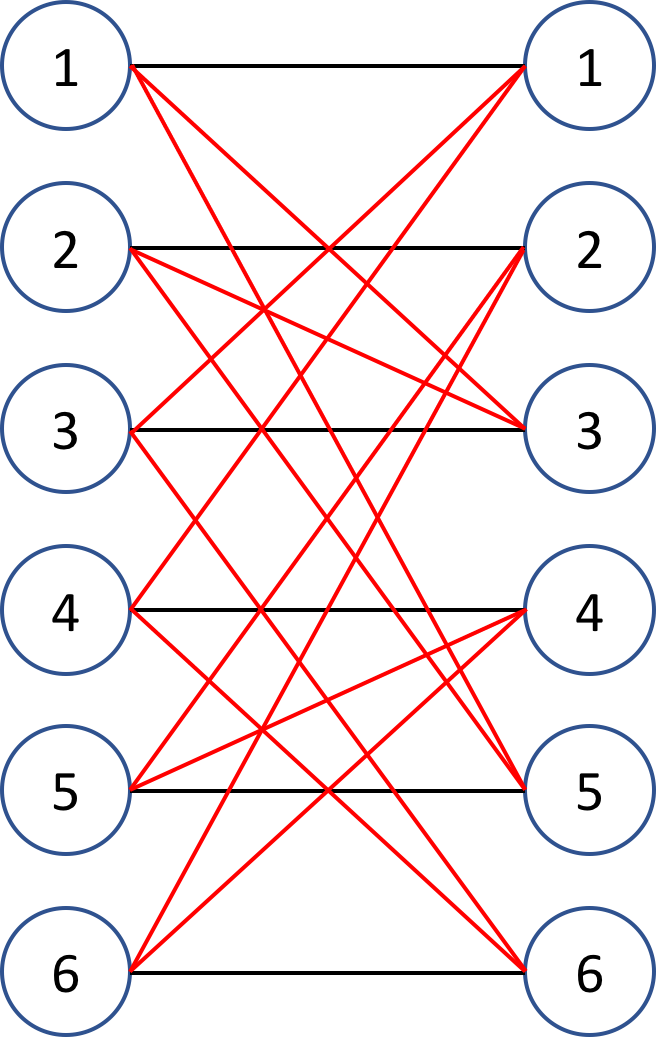}}
\hspace{13pt}
\subfigure[Alignment-conflict graph]{\includegraphics[width=0.28\linewidth]{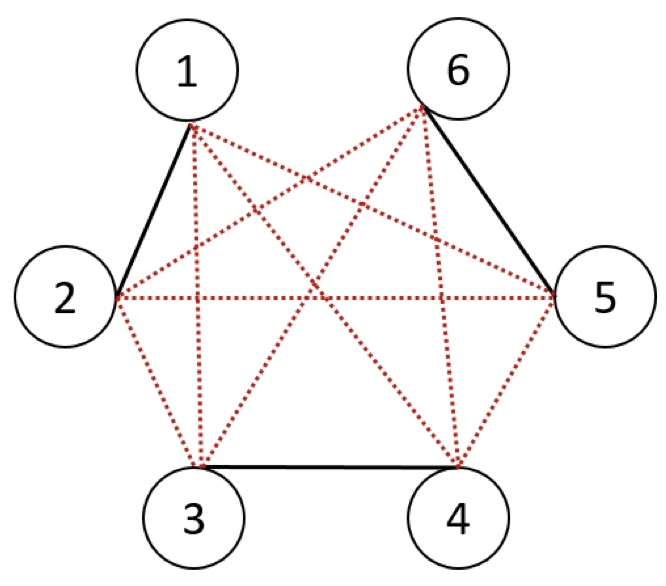}}
\caption{Maximal topology derived from 3-alliance construction for 6-user interference channel.}
\label{NAC}
\end{figure*}

\begin{example}
A sub-alliance graph is given in Fig. 3 (a).  A maximal topology from Fig. 3 (a) is given in Fig. 3 (b) and its alignment-conflict graph is given in Fig. 3 (c).

%there are maximal topology and related graphs derived from 3-alliance construction for 6-user interference channel. (a) is the alliance map and (b) is the alliance map with sub-alliance information. (c) is an maximal topology derived from (a). (d) is alignment and conflict graph of (c). 

Let $\mathcal{A}_{1}=\mathcal{{A}}_{1,2}\cup \mathcal{{A}}_{1,3}$, whose sub-alliances are $\mathcal{A}_{1,2}=\{W_1\}$ and $\mathcal{A}_{1,3}=\{W_2\}$, $\mathcal{A}_{2}=\mathcal{{A}}_{2,1}\cup \mathcal{{A}}_{2,3}$, whose sub-alliances are $\mathcal{A}_{2,1}=\{W_3\}$ and $\mathcal{A}_{2,3}=\{W_4\}$, and $\mathcal{A}_{3}=\mathcal{{A}}_{3,1}\cup \mathcal{{A}}_{3,2}$, whose sub-alliances are $\mathcal{A}_{3,1}=\{W_5\}$ and $\mathcal{A}_{3,2}=\{W_6\}$. Contrary to 2-alliance construction, messages in each alliance are partially interfered by all messages in each alliance and partitioned into sub-alliances indicating interferers. There are many ways to construct alliances by changing the number of alliances and partitioning messages into sub-alliances differently.
\end{example}

\subsection{Beamforming Vector Design for Alliance Construction}

In this subsection, we design a beamforming vector assigned for messages in each alliance and show that the linear beamforming scheme with alliance construction achieves the optimal symmetric DoF in TIM. Suppose that there are $N$ alliances with mutually partial hostility for $K$-user interference channel and we use two time extensions for beamforming vectors. Then $N$ pairwise linearly independent beamforming vectors can be constructed and each of them is allotted to each alliance. Let $\bfit{V}_{n}$ be a $2\times1$ beamforming vector for messages in alliance $\mathcal{A}_{n}$, $n\in\{1,2,\cdots,N\}$. There are no conflict among messages in an alliance and the messages in $\mathcal{A}_{n,m}$ are only interfered by all messages in $\mathcal{A}_{m}$. Consider the $i$th receiver that wants message $W_i$, which belongs to sub-alliance $\mathcal{A}_{n,m}$ after the alliance construction. Then the $2\times1$ received signal vector at receiver $i$ for two time slots is given as
%We have already assumed that channel coefficients and the topology are fixed throughout the duration of communication. 

\begin{equation}\label{bvd}
    \bm{Y}_{i}={\textit{h}_{ii}\bm{V_{n}}W_{i}}+\sum_{W_k\in \mathcal{A}_{m}}{\textit{h}_{ik}\bm{V_{m}}W_{k}}+\bm{{Z}_{i}}.
\end{equation}
Since $\bm{V}_{n}$ and $\bm{V}_{m}$ are linearly independent, receiver $i$ can null the aligned interference signals corresponding to messages in $\mathcal{A}_m$ and recover $W_i$. In the same way, every receiver can decode its desired message by only two time extensions, which means that the network achieves the optimal symmetric DoF 1/2 in TIM. Alliance construction relates messages in such a way that every message in each alliance must be interfered by all messages from only one alliance. The beamforming vectors split each received signal space into two subspaces with desired message and one directional aligned interference signals for all receivers

\subsection{Maximum Number of Alliances and Partition of Messages into Alliances}

%In previous Subsection, we prove that the maximal topology set is identified by the alliance construction. 
%For the completeness of the maximal topology set, 

\begin{figure*}[t]
\centering
\subfigure{\includegraphics[width=0.5\linewidth]{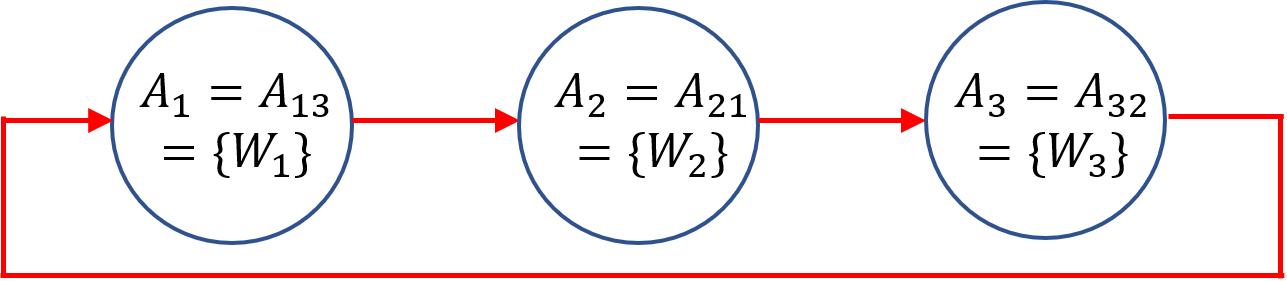}}
\caption{3-alliance construction requiring the minimum number of messages.}
\label{Partition}
\end{figure*}

In this subsection, we derive the maximum number of alliances for a given number of messages $K$ for the alliance construction. In fact, we derive the minimum number of messages required to construct $N$ alliances with mutual partial hostility rather than the maximum number of alliances that can be made with $K$ messages. Let $a_N$ be the minimum number of messages which can construct $N$ alliances with mutual partial hostility. It is trivial that $a_1=1$ and $a_2=2$. When $N=3$, the alliance construction requires the minimum number of messages, $K=3$, where there is only a message in each alliance and hostility between alliances is a tail-ending as in Fig.논문 \ref{Partition}.

%Even alliance construction can derive the maximal topology set, the maximum number of alliance is required to specify the range of the set.

It is clear that when all alliances are connected by only one-way hostility, the alliance construction requires the minimum number of messages. Every non-empty sub-alliance requires at least one message. And it is enough to make mutually partial hostility for each pair of alliances with only one non-empty sub-alliance of them. Assume that $N$ alliances have already been constructed using the minimum number of messages $a_N$ and we want to add a new alliance $\mathcal{A}_{N+1}$. In this situation, $\mathcal{A}_{N+1}$ should relate hostility with all existing $N$ alliances, which requires non-empty sub-alliance $\mathcal{A}_{i,N+1}$ or $\mathcal{A}_{N+1,i}$ for each $i \in \{1,\ldots,N\}$ and this requires at least $N$ additional messages. Thus, the recurrence relation is formulated as

\begin{equation}
    a_{N+1}=a_{N}+N, \hspace{10pt}N \geq 3
\end{equation}
and thus $a_{N}$ is computed as
\begin{equation}\label{mnu}
    a_{N}={N \choose 2},\hspace{10pt}N \geq 3.
\end{equation}
In fact, alliance construction with the minimum number of messages is equivalent to the handshake problem. Using (\ref{mnu}), the maximum number of alliances for given $K$ users can be derived as follows. Let $N$ be the maximum number of alliances for a given number of messages $K$. The range of $K$ which can construct alliances up to $N$ is given as
\begin{equation}
    \hspace{10pt}{N \choose 2} \leq K < {N+1 \choose 2}.
\end{equation}
It is clear that different alliance constructions are possible for the same number of alliances and messages because there are many way to partition messages into sub-alliances. The partition of messages for a given number of alliances $N$ is summarized as follows.

\begin{theorem}[Partition of messages]
There exist $N$ alliances with mutually partial hostility for $K$ user interference channel, if the number of messages in sub-alliance satisfies the following conditions for $K\geq{N \choose 2}$:
\begin{enumerate}[label=(\roman*)]
\item   $ \sum_{j=1, j\neq i}^{N}\vert\mathcal{A}_{j,i}\vert\geq 1$, for all $i$;
\item   $ \sum_{j=1, j\neq i}^{N}\vert\mathcal{A}_{i,j}\vert\geq 1$, for all $i$;
\item  $\vert\mathcal{A}_{i,j}\vert+\vert\mathcal{A}_{j,i}\vert\geq1$, for all distinct $i$ and $j$;
\item  $ \sum_{i=1}^{N}\sum_{j=1, j\neq i}^N\vert\mathcal{A}_{i,j}\vert=K$.
\end{enumerate}

\label{THM_PARTITION}
\end{theorem}

The first inequality constraints that every alliance has at least one common messages to conflict with. The second one constraints that every alliance has at least one messages. The third inequality is necessary and sufficient conditions of the mutual hostility between $\mathcal{A}_i$ and $\mathcal{A}_j$. The last one means that every message should belong to a sub-alliance. The condition for $K$ is required to ensure enough messages for constructing $N$ alliances. We omit the proof of the theorem.

%that is, every message should be interfered from all messages in an alliance. 

%\subsection{The Number of Conflicts in Alliance Construction}
%In this subsection, we derive the maximum number of alliances in alliance construction.

%\begin{theorem}
%Suppose that there are $N$ alliances with mutually partial hostility in $K$-user interference channel where alliance $\mathcal{A}_{i}$ has $a_i$ messages for $1\leq i \leq N$,
%$a_N\geq a_{N-1} \geq \cdots \geq a_1$ and $K\geq{N \choose 2}$. The maximum number of conflicts in alliance %construction $C_{\{ a_1,a_2,\cdots,a_N \}}$ is
%\begin{equation}
%C_{\{a_1,a_2,\cdots,a_N\}}=\sum_{k=2}^{N}a_k\{\sum_{j=1}^{k-1}a_j- {k-1 \choose 2} \}
%\end{equation}
%\end{theorem}

%\begin{proof}
%It is clear that the more conflicts occur when $\mathcal{A}_{ij}=\emptyset$ and $\mathcal{A}_{ji} \neq \emptyset$ for $i\geq j$. However, it is not possible to set $\mathcal{A}_{iN}=\mathcal{A}_{i}$ for $1\leq i \leq N-1$, because only $\mathcal{A}_{N}$ and $\mathcal{A}_{i}$ are mutually partial hostile and any two alliances except $\mathcal{A}_{N}$ are not mutually partial hostile. Thus, it is required to leave room for  the remain alliances with mutually partial hostility. Thus, $\mathcal{A}_{N}$ can use messages in the remain alliances up to $\sum_{i=1}^{N-1}a_i- {N-1 \choose 2}$. ${N-1 \choose 2}$ is the minimum number of messages for $N-1$ alliances with mutually partial hostility. Likewise, $\mathcal{A}_{i}$ can use messages in the remain alliances up to $\sum_{j=1}^{i-1}a_j- {i-1 \choose 2}$ and we prove it.
%\end{proof}

\subsection{Discriminant and Transformation of Maximal Topology}

In this subsection, discriminant of maximal topology is proposed using alliance construction and the transformation of non-maximal topology into maximal one is also proposed.

\begin{proposition}[Discriminant of maximal topology] 
The maximality of topology is determined as follows:

\begin{enumerate}[label=(\roman*)]
\item Construct all alignment sets (i.e., tentative alliances) for a given topology.
\item Investigate messages in each alignment set whether they follow the deconflict and cooperative conflict of alliance or not. If yes, alignment sets become alliances.
\item Investigate whether all alliances are pairwise mutually hostile or not, that is, there is no message which is not interfered and there is no pair of alliances $\mathcal{A}_{i}$ and $\mathcal{A}_{j}$ whose sub-alliances $\mathcal{A}_{i,j}$ and $\mathcal{A}_{j,i}$ are both empty sets. If yes, the topology is maximal.
\end{enumerate}

\end{proposition}

We propose how to transform non-maximal topology into maximal topology. Only the topology whose alignment sets have no internal conflict can be transformed into maximal topology by adding interference links properly. The reason why we consider disconnected interference links is that the links treated as disconnected actually exist and degenerate the SINR performance. And thus, it is desirable to consider more interference links without changing the optimality of DoF in TIM.

Before transformation, alignment sets (i.e, tentative alliances) are constructed such that each alignment set follows the deconflict of messages in the set. 

%After that, the maximality is achieved by checking whether each alignment set satisfies two conditions of alliance construction . 

\begin{proposition}[Transformation of non-maximal topology into maximal topology] 
Suppose that all alignment sets have no internal conflict.
\begin{enumerate}[label=(\roman*)]

%\item If there are messages that does not give any interference to other receivers, add interference links 
\item Add interference links so that all messages in each alignment set follow cooperative conflict of alliance, which makes all alignment sets into alliances.
\item If there is no hostility between $\mathcal{A}_{i}$ and $\mathcal{A}_{j}$, there are two ways to transform the topology as:
\begin{enumerate}
	\item Merge two alliances into a single one by combining $\mathcal{A}_{k,j}$ and $\mathcal{A}_{k,i}$ for distinct $i$, $j$, and $k$, where combining $\mathcal{A}_{k,j}$ and $\mathcal{A}_{k,i}$ requires to add conflict edges from all messages in $\mathcal{A}_{i}$ to each message in $\mathcal{A}_{k,j}$ and from all messages in $\mathcal{A}_{j}$ to each message in $\mathcal{A}_{k,i}$.
\item If there exist a message $W_n \in \mathcal{A}_{i}$ or $W_n \in \mathcal{A}_{j}$ whose receiver is not given any interference, add corresponding conflict edges from all messages in $\mathcal{A}_{j}$ to messages $W_n$ or from all messages in $\mathcal{A}_{i}$ to messages $W_n$.
\end{enumerate}

\item If there still exists a message in $\mathcal{A}_i$ whose receiver is not given any interference, add conflict edges from all messages in any alliances except $\mathcal{A}_i$ to the message.
\end{enumerate}
\end{proposition}

\begin{figure*}[t]
\centering
\subfigure[Original topology]{\includegraphics[width=0.25\linewidth]{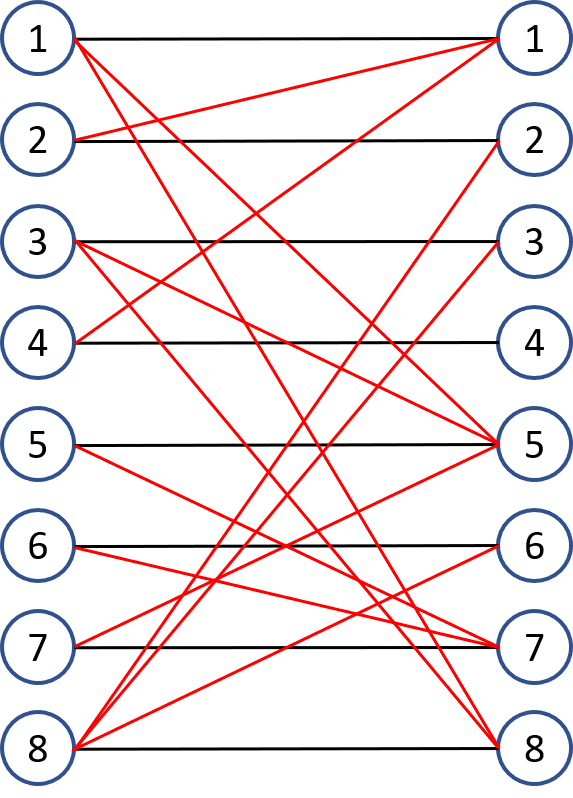}}
\hspace{13pt}
\subfigure[Maximal topology from (a)]{\includegraphics[width=0.25\linewidth]{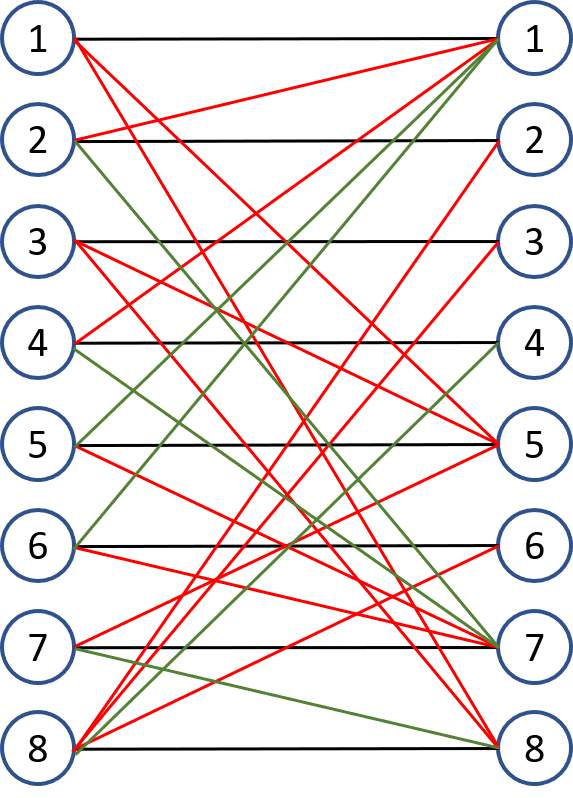}}
\hspace{13pt}
\subfigure[Other maximal topology from (a)]{\includegraphics[width=0.25\linewidth]{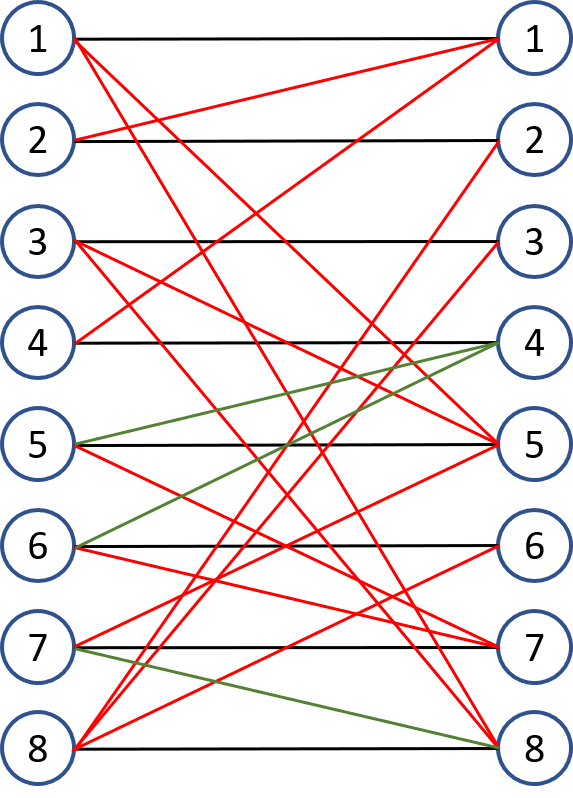}}
\caption{Transformation of non-maximal topology into maximal topologies for 8-user interference channel.}
\label{TF}
\end{figure*}

\begin{example}
A topology for 8-user interference  is given in Fig. \ref{TF}. There are four provisional alliances (i.e, alignment sets) with no internal conflict, $\mathcal{A}_{1}=\{W_1,W_3, W_7\}$ with sub-alliances $\mathcal{A}_{1,2}=\{W_1\}$, $\mathcal{A}_{1,3}=\{W_7\}$, and $\mathcal{A}_{1,4}=\{W_3\}$, $\mathcal{A}_{2}=\{W_2,W_4\}$ with sub-alliance $\mathcal{A}_{2,4}=\{W_2\}$, $\mathcal{A}_{3}=\{W_5,W_6\}$ with sub-alliances $\mathcal{A}_{3,1}=\{W_5\}$ and $\mathcal{A}_{3,4}=\{W_6\}$,  and  $\mathcal{A}_{4}=\mathcal{A}_{4,1}=\{W_8\}$. Since $\mathcal{A}_{1}$ does not follow the cooperative conflict of alliance, the interference link from transmitter 7 to receiver 8 should be added. Also, there is no hostility between $\mathcal{A}_2$ and $\mathcal{A}_3$ and $W_4$ in $\mathcal{A}_2$ is not interfered. For this situation, there are two ways to transform the topology into maximal one. The first one is to merge alliances $\mathcal{A}_2$ and $\mathcal{A}_3$ and add interference links from transmitters in an arbitrary alliance to receiver 4. It is also required to add interference links from transmitters in $\mathcal{A}_2$ to receivers in $\mathcal{A}_{1,3}$ and from transmitters in $\mathcal{A}_3$ to receivers in $\mathcal{A}_{1,2}$. A maximal topology is given in Fig. \ref{TF} (b). Another way is to make hostility between $\mathcal{A}_2$ and $\mathcal{A}_3$ by setting $\mathcal{A}_{2,3}=\{W_4\}$. Other maximal topology is given in Fig. \ref{TF} (c). There are many ways to transform non-maximal topology into maximal topology by adding interference links.
\end{example}

\section{Topology Matrix}
In the previous section, the conditions for maximal topology were specified as the relationship of messages in alliance and relationship among alliances. However, it is hard to determine maximality of topology in the topology and message graph. In this section, the  analysis of topology in the matrix form is proposed to derive an MTM, determine the maximality of topology matrix, and transform non-MTM into MTM.

\subsection{Maximal Topology Matrix}

In this subsection, the  sufficient and necessary conditions of MTM are delivered  based on alliance construction. First, some of definitions related to topology matrix are givne as follows.

\begin{definition}[Alliance block and interference block]
Suppose that the indices of messages in each alliance are ordered consecutively as $\mathcal{A}_{n}=\{W_i , W_{i+1},\cdots,W_{i+\vert \mathcal{A}_{n}\vert-1}\}$. A principal submatrix of $\bfit{T}$ is called  an alliance block of $\mathcal{A}_{n}$ if $\bfit{T}$ satisfies the following conditions:
\begin{enumerate}[label=(\roman*)]
\item The principal submatrix corresponding to $\mathcal{A}_{n}$ is an identity matrix of size $\vert \mathcal{A}_{n}\vert \times \vert \mathcal{A}_{n}\vert$.
\item There exist at least one $j$ in $\bfit{T}$ such that $t_{kj}=1$ for all $k\in \{i,i+1,\cdots,i+\vert \mathcal{A}_n \vert -1\}$ and $W_j\notin \mathcal{A}_{n}$.
\item There does not exist $j$ in $\bfit{T}$ such that $t_{kj}=1$ and $t_{lj}=0$ for some $k,l \in \{i,i+1,\cdots,i+\vert \mathcal{A}_n \vert -1\}$ and $W_j\notin \mathcal{A}_{n}$.

\end{enumerate}

The above $\vert \mathcal{A}_{n} \vert \times 1$ submatrix $[t_{kj}]$ is called an interference block from $\mathcal{A}_{n}$.
\end{definition}

The first condition corresponds to the deconflict of messages in alliance preventing internal conflict and the second and third ones are the cooperative conflict of alliance.

\begin{figure*}[t]
\centering
\subfigure[Topology]{\includegraphics[width=0.18\linewidth]{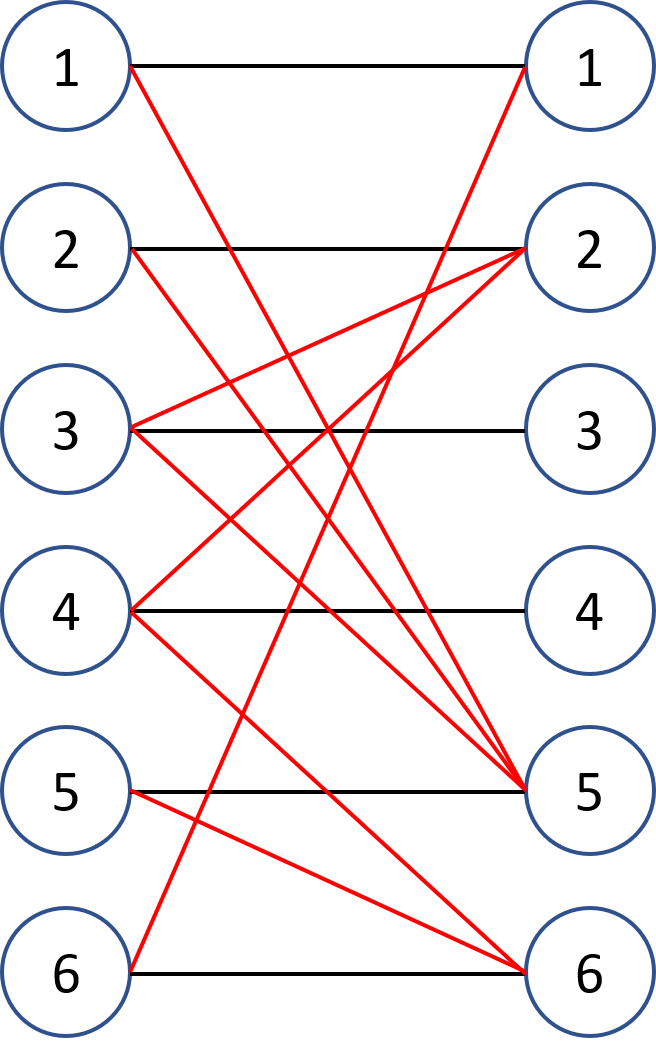}}
\hspace{25pt}
\subfigure[Topology matrix]{\includegraphics[width=0.27\linewidth]{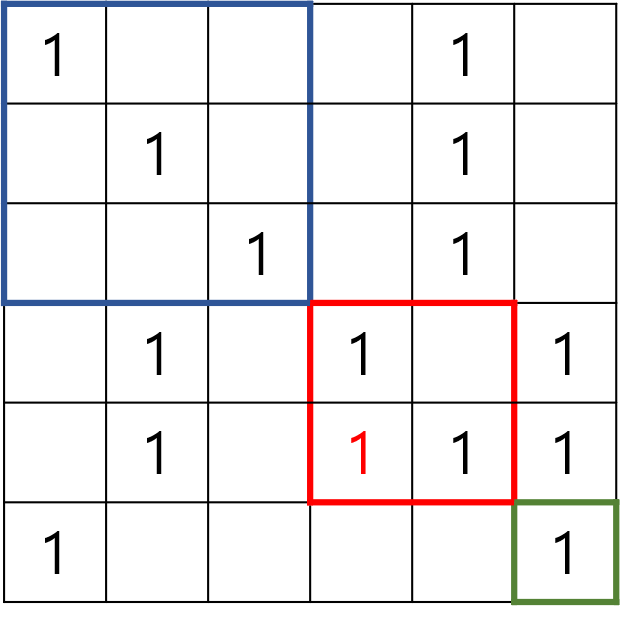}}
\caption{Topology and its topology matrix for 6-user interference channel.}
\label{AB}
\end{figure*}

\begin{example}
Fig. \ref{AB} shows topology and its topology matrix for 6-user interference channel. There are three principal submatrices whose sizes are $3\times3$, $2\times 2$, and $1\times 1$. The $3\times3$ and $1\times1$ square matrices are alliance blocks. But $2\times 2$ one is not due to $t_{5,4}=1$, which represents internal conflict in alignment set $\{W_4, W_5\}$. There is an interference block from the alliance block $\mathcal{A}_{1}=\{W_1,W_2,W_3\}$  at $t_{i,5}$, $i\in\{1,2,3\}$, which represents the cooperative conflict to $W_5$. Thus the above matrix is not MTM.
\end{example}

It is also required to translate the mutually partial hostility in alliance construction into topology matrix for MTM.

\begin{corollary}[MTM]
 Suppose that there are $N$ alliance blocks in a topology matrix and messages in each alliance are ordered consecutively in indices. A topology matrix is MTM if and only if all alliance blocks satisfy following conditions:

 \begin{enumerate}[label=(\roman*)]
\item Each column of the alliance blocks has a single interference block.
\item There exist interference block(s) between any two alliance blocks.
 \end{enumerate}
\end{corollary}

The first condition ensures that there is no message which is not interfered and every message is interfered from all messages in an alliance. The second condition ensures that at least one of sub-alliances $\mathcal{A}_{i,j}$ and  $\mathcal{A}_{j,i}$ are not empty-set for any $\mathcal{A}_i$ and $\mathcal{A}_{j}$. In fact, the above two conditions correspond to the mutually partial hostility for the alliance construction. Thus we omit the proof of Corollary 2.

\begin{figure*}[t]
\centering
\subfigure[Maximal topology]{\includegraphics[width=0.3\linewidth]{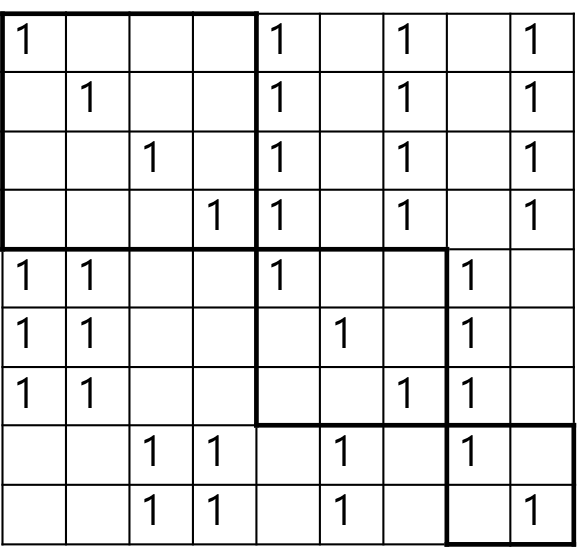}}
\hspace{25pt}
\subfigure[Non-maximal topology ]{\includegraphics[width=0.3\linewidth]{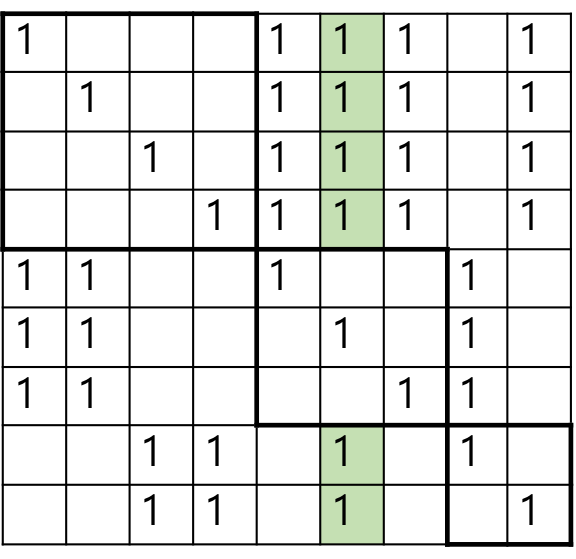}}
\caption{Topology matrices for 9-user interference channel.}
\label{MTMF}
\end{figure*}

\begin{example}
In Fig. \ref{MTMF}, there are two topology matrices for 9-user interference channel. The topology matrix in Fig. 7 (a) satisfies all conditions of MTM. On the other hand, the topology matrix in Fig. 7 (b) is not MTM because there is a message $W_6$ in $\mathcal{A}_2$ which has two interference blocks from $\mathcal{A}_1=\{W_1,W_2,W_3,W_4\}$ and  $\mathcal{A}_3=\{W_8,W_9\}$.
\end{example}

\subsection{Discriminant and Transformation of MTM}

In this subsection, we propose the discriminant of MTM in matrix perspective. The interpretation of the discriminant of maximal topology into the topology matrix is needed because the characteristics of maximal topology are more easily analyzed in topology matrix than sub-alliance graph. We assume that the messages that belong to the same alliance are ordered consecutively in indices. However, the indices of messages in a given topology matrix are always not well sorted and the alliance and interference blocks are not easily discerned. Thus, the permutation of messages for consecutive ordering of indices for each alliance should precede the analysis of maximality of topology in matrix perspective.
%The message indices are permuted in such a way that the messages in the same alliance are ordered consecutively.

\begin{definition}[Permutation of matrix]
Let $[\bfit{A}]_{i \leftrightarrow j}$ represent column and row permutations of matrix indices that swap the $i$th row and column with the $j$th ones, which corresponds to exchanging indices of two messages $W_i$ and $W_j$.
\end{definition}

%$[\bfit{A}]_{i \leftrightarrow j}$ means that the message $W_i$ and $W_j$ exchange their indices. 

\begin{proposition}[Discriminant of MTM]
The maximality of topology matrix is determined as follows:
\begin{enumerate}[label=(\roman*)]
\item Construct all tentative alliance blocks by permutating matrix indices in a such way that any $i$th and $j$th rows and $i$th and $j$th columns are simultaneously ordered consecutively if $t_{i,k}=t_{j,k}=1$.
\item Investigate whether all tentative alliance blocks are alliance blocks or not. If not, it is not MTM.
\item If yes, investigate whether all alliance blocks  follow two conditions in Corollary 2 or not.
\item If yes, it is an MTM.
\end{enumerate}

\end{proposition}

\begin{figure*}[t]
\centering
\subfigure[Original matrix]{\includegraphics[width=0.27\linewidth]{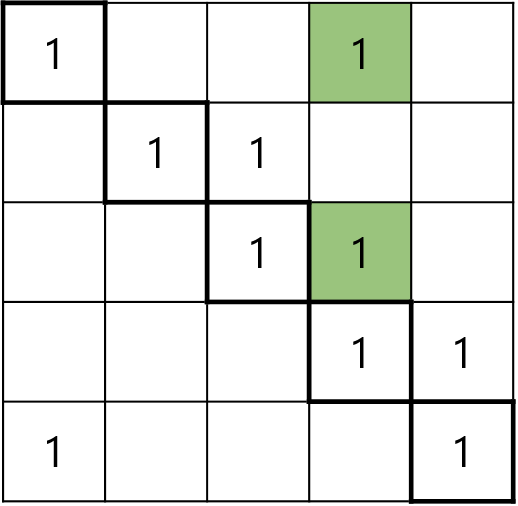}}
\hspace{15pt}
\subfigure[Matrix with proper tentative alliance blocks]{\includegraphics[width=0.27\linewidth]{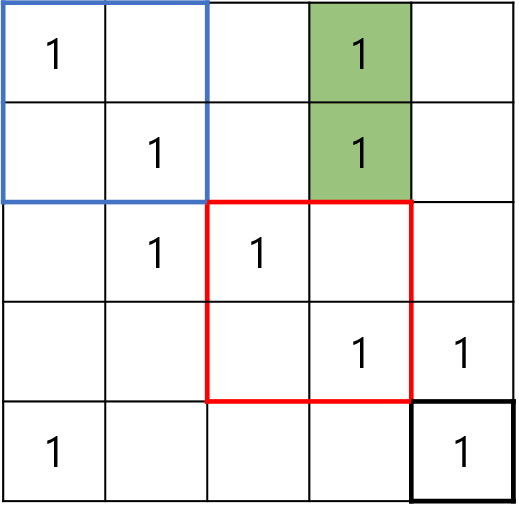}}
\hspace{15pt}
\subfigure[Matrix after transformation]{\includegraphics[width=0.27\linewidth]{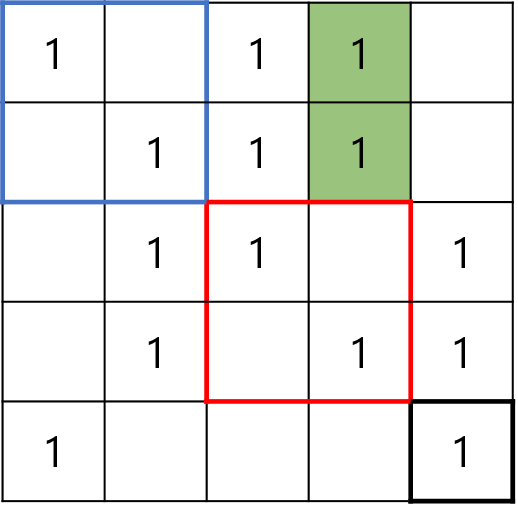}}
\caption{Topology matrices for 5-user interference channel.}
\end{figure*}

\begin{example}
The topology matrix in Fig. 8 represents a 5-user interference channel. However, if we swap message indices 2 and 3 in both columns and rows, $\mathcal{A}_{1}=\{W_1\}$ and $\mathcal{A}_{3}=\{W_3\}$ can be combined into a single alliance block because $t_{1,4}=t_{3,4}=1$. Even though message indices are reordered, this topology matrix is still not an MTM, because every column in the matrix does not have one interference block. Thus it is possible to add more interference links while maintaining the DoF performance. It is not trivial to determine which empty spaces should be filled with element 1 in the topology matrix. We propose how to make MTM from non-MTM by filling some empty spaces with element 1 as in the following proposition.
\end{example}

\begin{proposition}[Transformation of non-MTM into MTM]
First, check whether each principal submatrix is an identity matrix or not. If yes, the transformation can be stated as:

\begin{enumerate}[label=(\roman*)]
\item Insert element 1 to the topology matrix in a such way that incomplete interference blocks do not exist.
\item If two alliance blocks $\mathcal{A}_n$ and $\mathcal{A}_m$ do not have any corresponding interference block, there are two ways to transform the topology matrix as:
\begin{enumerate}
    \item Merge them by permutating matrix indices in a such way that all indices in $\mathcal{A}_n$ and $\mathcal{A}_m$ are rearranged consecutively.
    
    \item If there exist the $i$th column with no interference block for message $W_i\in \mathcal{A}_{n}$ or $W_i\in \mathcal{A}_{m}$, add corresponding interference block to the $i$th column of $\mathcal{A}_{n}$ or $\mathcal{A}_{m}$.

\end{enumerate}

\item If there still exists a column with no interference block, add an arbitrary interference block to the column.
\end{enumerate}
\end{proposition}

Proposition 5 shows that transformation is not unique for a given topology matrix. There are many ways to merge provisional alliance blocks into single alliance block. Also for the column with no interference block, there are many ways to put an arbitrary interference block into the column.

\section{Generalized Alliance Construction}

\vspace{10pt}
\subsection{Generalized Alliance Construction}
\vspace{10pt}

Until now, we focus on alliance construction for maximal topology and analyze characteristics of alliance construction and its topology matrix. But in TIM, it is also possible to achieve DoF less than 1/2, but it is more difficult to analyze it. In this section, we propose alliance construction for DoF less than $1/2$ by modifying the definition of  sub-alliance and derive its topology.

%First we briefly introduce  the types of alliances and deliver definition. There are two types of alliances. One is true alliance we have already explained. True alliance is a set of messages which follow the deconflict of messages in alliance and the coopeartive conflict of alliance. Because they deconflict each other, we call it true alliance. 
%Fake alliance is also a set of messages that attack same enemy messages but some of them attack each other. Therefore, there must exist internal conflict among these messages and topologies for these can not achieve optimal DoF in TIM. However, it is possible to achieve DoF less than 1/2 by assigning independent beamforming vectors to each fake alliance. 

\begin{definition}[Generalized sub-alliance]
The alliance $\mathcal{A}_{i}$ is partitioned into $n_i$ generalized sub-alliances $\mathcal{A}_{i,{{\mathcal{E}}_{i}^{k}}}$, where  $\mathcal{E}_{i}^{k}$ is the set of indices of alliances whose messages give interference to all messages in $\mathcal{A}_{i,{{\mathcal{E}}_{i}^{k}}}$, $\cup_{k=1}^{n_i}\mathcal{A}_{i,{{\mathcal{E}}_{i}^{k}}}=\mathcal{A}_{i}$ and $\mathcal{E}_{i}^{k}$ is not a subset of the others but $\mathcal{E}_{i}^{k_1}\cap \mathcal{E}_{i}^{k_2}\neq \emptyset$.

\end{definition}

\begin{definition}[Multiple partial hostility]
The alliances $\mathcal{A}_{i}$ and $\mathcal{A}_{j}$ are multiple partial hostile if $j\in \cup_{k=1}^{n_i}\mathcal{E}_{i}^{k}$ and/or $i\in \cup_{k=1}^{n_j}\mathcal{E}_{j}^{k}$. 
\end{definition}

\begin{theorem}[Generalized symmetric DoF]
Suppose that there are alliances with generalized sub-alliances and multiple partial hostility. 
Let $E_M$ be $\max_{i,k}{\left| \mathcal{E}_{i}^{k} \right|}$ in TIM of the interference channel. The achievable symmetric DoF using the proposed linear beamforming scheme is 
\begin{equation}
    d_{sym}=\frac{1}{E_M+1}. 
\end{equation}
\end{theorem}

\begin{proof}

(Achievability) Suppose that there are $N$ alliances with generalized sub-alliances and multiple partial hostility for $K$-user interference channel and we use $(E_M+1)$ time extensions for beamforming vectors. It is possible to construct $N$ beamforming vectors allotted to each alliance, where any $E_M+1$ vectors in $N$ vectors are linearly independent. Let $\bfit{V}_{n}$ be an $(E_M+1)\times1$ beamforming vector for messages in $\mathcal{A}_{n}$, $n\in\{1,2,\cdots,N\}$. There is no conflict among messages in each alliance and each message in $\mathcal{A}_{n,\mathcal{E}_{n}^{k}}$ are interfered by all messages in all alliances $\mathcal{A}_m$, $m\in \mathcal{E}_n^{k}$. 

Consider the $i$th receiver that wants message $W_i$, which belongs to $\mathcal{A}_{n,\mathcal{E}_{n}^{j}}$  after the alliance construction. Then the $(E_M+1) \times1$ received signal vector at receiver $i$ for $(E_M+1)$ time slots is given as
\begin{equation}\label{bvd}
    \bm{Y}_{i}={\textit{h}_{ii}\bm{V_{n}}W_{i}}+\sum_{m\in \mathcal{E}_n^k}  \sum_{W_j\in \mathcal{A}_{m}}{\textit{h}_{ij}\bm{V_{m}}W_{j}}+\bm{{Z}_{i}}.
\end{equation}
Since there are at most $E_M$ alliances with indices in $\mathcal{E}_{n}^{k}$ and any $E_M+1$ beamforming vectors are linearly independent, receiver $i$ can null the aligned interference signals and recover $W_i$. In the same way, every receiver can decode its desired message by only $E_M+1$ time extensions, which means that the interference channel achieves DoF 1/$(E_M+1)$ in TIM. 

(Upperbound) According to Theorem 5 in \cite{TIM}, the symmetric DoF in TIM is bounded as
\begin{equation}
d_{sym} \leq \frac{1}{\Psi},
\end{equation}

\noindent where $\Psi$ is the maximum cardinality of an acyclic subset of messages in demand graph. We just show that the maximum cardinality of an acyclic subset of messages is equal to $E_M+1$ in the proposed generalized alliance construction. Consider a message $W_p$ that belongs to a sub-alliance $\mathcal{A}_{i,\mathcal{E}_{i}^{k}}$, where $\left| \mathcal{E}_{i}^{k} \right|=E_M$. Then, there is no edge from receiver $p$ to transmitters of messages that belong to $E_M$ alliances with indices in $\mathcal{E}_{i}^{k}$ in demand graph. Also, the alliances with indices in $\mathcal{E}_{i}^{k}$ are multiple partial hostility, that is, for any $\mathcal{A}_n$ and $\mathcal{A}_m$ with indices in $\mathcal{E}_{i}^{k}$, there exist sub-alliances $\mathcal{A}_{n,\mathcal{E}_{n}^{\alpha_n}}$ and/or $\mathcal{A}_{m,\mathcal{E}_{m}^{\alpha_m}}$ where $m \in \mathcal{E}_{n}^{\alpha_n}$ and/or $n \in \mathcal{E}_{m}^{\alpha_m}$. Thus, if we choose $E_M$ tuple messages associated with $E_M$ alliances with indices in $\mathcal{E}_{i}^{k}$ having multiple partial hositliy with each other, there is no cycle for $W_p$ and $E_M$ messages in demand graph, which means that $\Psi=E_M+1$.
\end{proof}

The symmetric DoF achieved by linear beamforming scheme is bounded by the maximum number of interfering alliances for all generalized sub-alliances. Note that the interference channel can achieve the optimal symmetric DoF when each message of each sub-alliance in the interference channel are interfered by all messages from a single alliance, that is, $E_M=1$, which results in $d_{sym}=1/2$. 
%The following corollary can be stated without proof.

\begin{corollary}[Generalized symmetric DoF]
Suppose that there are alliances with generalized sub-alliances and multiple partial hostility in TIM of the interference channel. The topology derived from the generalized alliance construction is maximal for DoF $1/(E_M+1)$ if $\left| \mathcal{E}_{i}^{k} \right|=E_M$ for all $1\leq i\leq N$ and $1 \leq k \leq n_i$.
\end{corollary}

\begin{proof}
Since $\left| \mathcal{E}_{i}^{k} \right|=E_M$ for all $1\leq i\leq N$, if we add any interference link into the topology, $\max_{i,k}{\left| \mathcal{E}_{i}^{k} \right|}=E_M+1$ and the symmetric DoF in TIM, $d_{sym}=1/(E_M+2)$. Thus, the topology is maximal for DoF $1/(E_M+1)$.
\end{proof}

\subsection{Topology Matrix for Generalized Alliance Construction}

Generalized alliance encompasses interference with not only optimal DoF but also non-optimal DoF by generalizing sub-alliances and thus it is necessary to redefine MTM because the previous MTM only considers interference channels achieving the optimal DoF $1/2$. A maximal topology matrix for interference channels achieving symmetric DoF up to $1/n$ is called  MTM for DoF $1/n, n\geq 3$.

\begin{figure*}[t]

\centering
\subfigure[Non-MTM]{\includegraphics[width=0.25\linewidth]{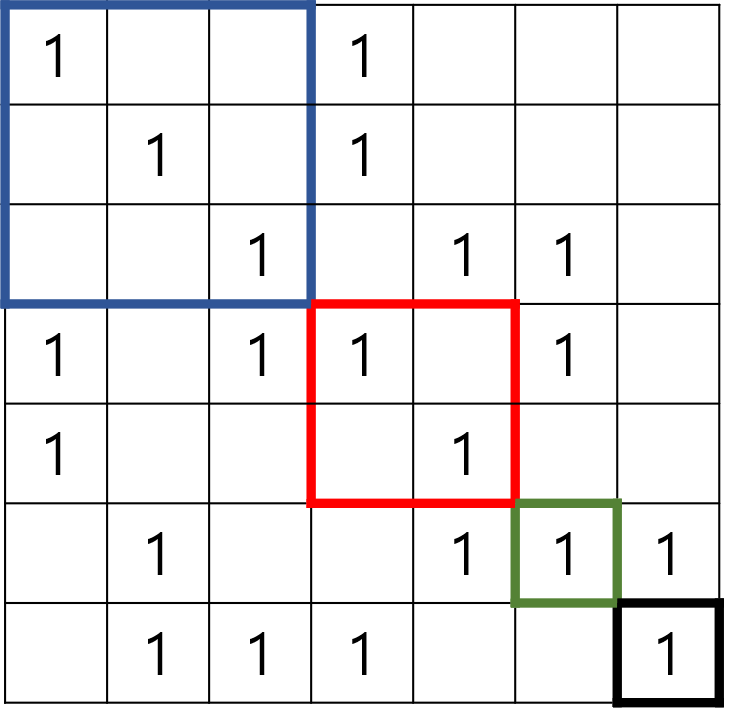}}
\hspace{20pt}
\subfigure[MTM]{\includegraphics[width=0.25\linewidth]{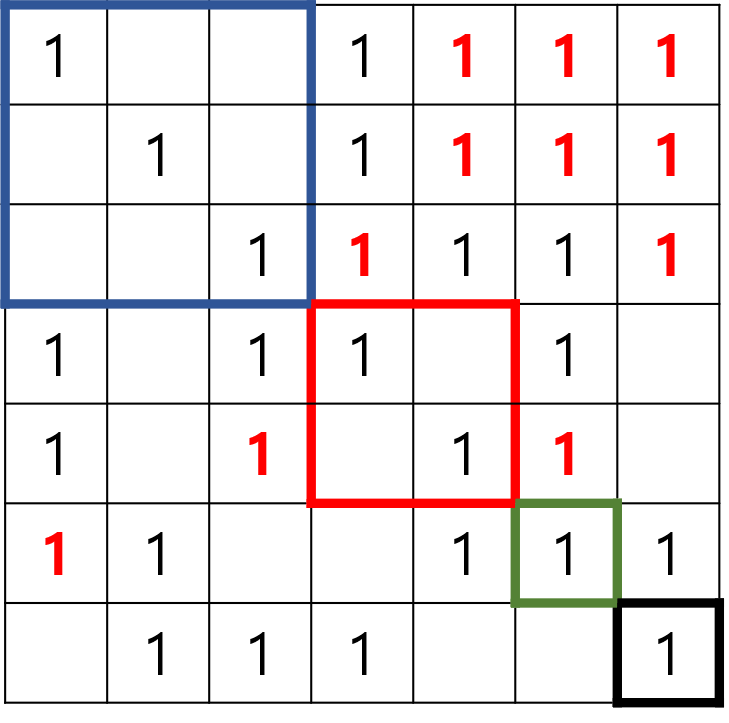}}

\caption{Topology matrices for 7-user interference channel achieving DoF $\frac{1}{3}$.}
\label{GMTM}
\end{figure*}

The following corollary is the matrix version of Corollary 3.

\begin{corollary}[MTM for DoF $1/n$]
 Suppose that there are $N$ alliance blocks in a topology matrix and messages in each alliance are ordered consecutively in indices. A topology matrix is MTM for DoF $1/(1+E_M)$ if all alliance blocks satisfy following conditions:

 \begin{enumerate}
\item Each column has $E_M$ interference blocks from $E_M$ alliances.
\item There exist at least one interference blocks between any two alliance blocks.
 \end{enumerate}
\end{corollary}
 
%We also mention the transformation non-MTM to MTM. After matrix indices permutation, the transformation is just to assign $n$ interference blocks to each column and thus there are many way to make non-MTM to MTM. 

\begin{example}
In Fig. \ref{GMTM}, there are two topology matrices for 7-user interference channel, which  have been already well permutated. Two matrices have four alliance blocks, respectively and both matrices can achieve symmetric DoF 1/3 in TIM, because $E_M$ is equal to 2. However, the topology matrix in Fig. 9 (a) is not MTM because there are lots of rooms for additional interference links. The topology matrix in Fig. 9 (b) is designed as an example of MTM from the topology matrix in Fig. 9 (a). The bold elements are inserted properly to satisfy the maximality of topology in Fig. 9 (b). After transformation, it can be seen that the topology matrix in Fig. 9 (b) satisfies two conditions in Corollary 4 and thus, it is an MTM for DoF 1/3.

%there is no incomplete interference block and every column has exact two interference blocks. 
\end{example}

\section{Conclusion}
In this paper, we introduced alliance as a set of messages that follows the deconflict of messages and the coopeartive conflict. Based on alliance, we proposed the alliance construction, which constructs and relates alliances with mutually partial hostility and generates maximal topology. Properties of alliance construction were given and the discriminant and transformation for maximal topology were also proposed. Moreover, we convert alliance construction based on message graph into topology matrix in order to analyze the maximality of topology easily. The sufficient and necessary conditions for MTM was delivered and the discriminant of MTM and the transformation of non-MTM into MTM were also proposed. Furthermore, we generalized the alliance construction with generalized sub-alliances dealing with interference channels for DoF 1/n. The generalized alliance construction was represented in matrix form and the conditions of MTM for DoF $\frac{1}{n}$ were described.


\begin{thebibliography}{99}

\bibitem{IA1} V. Cadambe and S. A. Jafar, "Interference alignment and the degrees of freedom of the K user interference channel," {\em IEEE Trans. Inf. Theory}, vol. 54, no. 8, pp. 3425-3441, Aug. 2008.

\bibitem{IA2} K. Gomadam, V. Cadambe, and S. A. Jafar, "A distributed numerical approach to interference alignment and applications to wireless interference networks," {\em IEEE Trans. Inf. Theory}, vol. 57, no. 6, pp. 3309-3320, Jun. 2011.

\bibitem{IA3} B. Nazer, M. Gastpa, S. A. Jafar, and S. Vishwanath "Ergodic interference alignment," {\em IEEE Trans. Inf. Theory}, vol. 58, no. 10, pp. 6355-6371, Oct. 2012.

\bibitem{BIA} T. Gou, C. Wang, and S. A. Jafar, "Aiming perfectly in the dark-blind interference alignment through staggered antenna switching," {\em IEEE Trans. Signal Process.}, vol. 59, no. 6, pp. 2734-2744, Jun. 2011.

\bibitem{BIA2} S. A. Jafar, "Blind interference alignment," {\em IEEE Trans. J. Sel. Topics Signal Process.}, vol. 6, no. 3, pp. 216-227, Jun. 2012.

\bibitem{TIM} S. A. Jafar, "Topological interference management through index coding," {\em IEEE Trans. Inf. Theory}, vol. 60, no. 1, pp. 529-568, Jan. 2014.

\bibitem{TIM_FF} N. Naderializadeh, "Interference networks with no CSIT: Impact of topology," {\em IEEE Trans. Inf. Theory}, vol. 61, no. 2, pp. 917-938, Feb. 2015.

\bibitem{TIM_ALT} H. Sun, C. Gen, and S. A. Jafar, "Topological interference management with alternating connectivity," {\em in Proc. IEEE Int. Symp. Inf. Theory (ISIT)}, Jul. 2013, pp. 399-403.

\bibitem{TIM_MA} H. Sun and S. A. Jafar, "Topological interference management with multiple antennas," {\em in Proc. IEEE Int. Symp. Inf. Theory (ISIT)}, Jun. 2014, pp. 1767-1771.

\bibitem{TIM_HEXA} Y. Gao, G. Wang, and S. A. Jafar, "Topological interference management for hexagonal cellular netowrks," {\em IEEE Trans. Wirelss Commun.}, vol. 14, no. 5, pp. 2368-2376, May 2015.

\bibitem{TIM_MP} X. Yi and G. Caire, "Topological interference management with decoded message passing," {\em IEEE Trans. Inf. Theory}, vol. 64, no. 5, pp. 3842-3864, May 2015.




\end{thebibliography}
\end{document}